\documentclass{comnet}
\usepackage{amsmath} 
\usepackage{amsthm} 

\usepackage{graphicx} 

\usepackage{bibentry}

\usepackage{amssymb}

\usepackage{subcaption}
\usepackage{mleftright} 
\mleftright 

\usepackage{bm}

\graphicspath{{./fig/}}

\newcommand{\tc}{\tilde{c}}

\newcommand{\prw}{p^{\mathrm{rw}}}

\newcommand{\expC}{\langle \tilde{c} \rangle}
\newcommand{\expCst}{\langle \tilde{c} \rangle_{st}}
\newcommand{\expCts}{\langle \tilde{c} \rangle_{ts}}

\newcommand{\dist}{\mathit{\Delta}}

\newcommand{\giv}{\ensuremath{\: \vert \:}}

\newcommand{\Prob}{\mathrm{P}_{st}}
\newcommand{\ProbAll}{\mathrm{P}}
\newcommand{\ProbRSP}{\mathrm{P}_{st}^{\mathrm{RSP}}}

\newcommand{\Prw}{\mathrm{P}_{st}^{\mathrm{rw}}}

\newcommand{\PRW}{\mathbf{P}^{\mathrm{rw}}}

\newcommand{\Pset}{\mathcal{P}_{st}}

\newcommand{\PF}{\mathcal{Z}_{st}}

\newcommand{\Lhood}{\mathcal{L}}

\newcommand{\logL}{\log \mathcal{L}}

\newcommand{\V}{n}
\newcommand{\Vij}{n_{ij}}
\newcommand{\Vi}{n_{i}}
\newcommand{\expVij}{\bar{n}_{ij}}

\newcommand{\expVi}{\bar{n}_{i}}

\newcommand{\psum}{\sum_{\wp \in \Pset}}
\newcommand{\Psum}{\sum\limits_{\wp\in \Pset}}

\newcommand{\betaMLE}{\hat{\beta}_{\mathrm{MLE}}}

\newcommand{\relent}{\mathbb{J} \left( \Prob \| \Prw \right)}

\newcommand{\tw}{\widetilde{w}}

\newcommand{\Wt}{\overset{-t}{\mathbf{W}}\vphantom{\mathbf{W}}}

\usepackage[T1]{fontenc}

\DeclareMathAlphabet{\mathcal}{OMS}{cmsy}{m}{n}
\DeclareSymbolFont{symbols2}{OML}{cmm}{m}{it}
\DeclareMathSymbol{\wp}{\mathord}{symbols2}{'175}

\usepackage{hyperref}

\AtBeginDocument{
  \label{CorrectFirstPageLabel}
  
}

\begin{document}

\title{Maximum likelihood estimation for randomized shortest paths with trajectory data}

\shorttitle{Maximum likelihood estimation for randomized shortest paths} 
\shortauthorlist{Ilkka Kivim\"aki et al.} 

\author{
\name{Ilkka Kivim\"aki$^*$}
\address{Aalto University, Department of Computer Science, Espoo, Finland\\
Universit\'{e} catholique de Louvain, ICTEAM, Louvain-la-Neuve, Belgium
\email{$^*$Corresponding author: ilkka.s.kivimaki@gmail.com}}
\name{Bram Van Moorter, Manuela Panzacchi}
\address{Norwegian Institute for Nature Research, Trondheim, Norway}
\name{Jari Saram\"{a}ki}
\address{Aalto University, Department of Computer Science, Espoo, Finland}
\and
\name{Marco Saerens}
\address{Universit\'{e} catholique de Louvain, ICTEAM, Louvain-la-Neuve, Belgium}
}

\maketitle

\begin{abstract}
{
Randomized shortest paths (RSP) are a tool developed in recent years for different graph and network analysis applications, such as modelling movement or flow in networks.
In essence, the RSP framework considers the temperature-dependent Gibbs-Boltzmann distribution over paths in the network.
At low temperatures, the distribution focuses solely on the shortest or least-cost paths, while with increasing temperature, the distribution spreads over random walks on the network.
Many relevant quantities can be computed conveniently from this distribution, and these often generalize traditional network measures in a sensible way.
However, when modelling real phenomena with RSPs, one needs a principled way of estimating the parameters from data.
In this work, we develop methods for computing the maximum likelihood estimate of the model parameters, with focus on the temperature parameter, when modelling phenomena based on movement, flow, or spreading processes.
We test the validity of the derived methods with trajectories generated on artificial networks as well as with real data on the movement of wild reindeer in a geographic landscape, used for estimating the degree of randomness in the movement of the animals.
These examples demonstrate the attractiveness of the RSP framework as a generic model to be used in diverse applications.
}
{randomized shortest paths, random walk, shortest path, parameter estimation, maximum likelihood, animal movement modelling}
\end{abstract}

\section{Introduction}
\label{sec:MLEIntro}

\subsection{Background and motivation}
\label{sec:BackgroundAndMotivation}
Modelling the movement and flow of different entities on networks is a key topic in network science \citep{wasserman1994social, kolaczyk2009statistical, newman2010networks, estrada2012structure, barabasi2015network, fouss2016algorithms}.
It is not only relevant for studying physical movement, but also for more abstract processes such as communication, the spread of diseases and financial transactions, to name a few.
Moreover, models of movement and flow are often used as the basis of popular distance, centrality and other measures on networks.

The two most standard paradigms for studying movement on networks consider it as occurring over shortest paths or over random walks.
These paradigms are, however, often too simple for building realistic models.
In reality, movement or flow rarely strictly follows the shortest paths, nor is it completely random.
In addition, the standard network measures derived from these paradigms often have caveats.
The shortest path distance, for instance, does not take into account connectivity besides the shortest path, i.e.\ the number of other short connections between nodes. 
Moreover, when comparing distances on unweighted graphs, the shortest path distance often results in ties between node pairs, because it typically yields a limited number of integer values.

Random walks can be used for defining the \emph{commute time} and \emph{commute cost distances} \citep{chandra1989the-electrical,fouss2016algorithms,kivimaki2014developments}, as well as the closely related \emph{resistance distance}~\citep{klein1993resistance}.
They reflect the connectivity between nodes, but only in small networks.
Instead, in large networks they become only dependent on the local connectivity of the nodes, in many cases the node degrees~\citep{luxburg2010getting, luxburg2014hitting}.
This phenomenon has been called the \emph{global information loss problem} in \citep{nguyen2016new,odor2017global}.

The above reasons have motivated the development of alternatives for the traditional paradigms, with special focus on defining new distances on networks.
One such alternative, which is also the main topic of this work, is the \emph{randomized shortest paths (RSP) framework} \citep{yen2008a-family, saerens2009randomized, kivimaki2014developments, kivimaki2016two}. 
The RSP framework is based on considering a Gibbs-Boltzmann distribution over paths from a \emph{source node} $s$ to a \emph{target node} $t$ involving an \emph{inverse temperature parameter}, $\beta=1/T > 0$, which controls the degree of randomization from the optimal, shortest paths.
At the limits of the parameter range, the distribution focuses either on solely the optimal paths ($\beta \rightarrow \infty$) or spreads over random walks ($\beta \rightarrow 0^{+}$).
In the framework, the optimality of a path can be defined as the number of steps, or, more generally, based on real-valued \emph{edge costs}, denoted by $c_{ij}$ for edge $(i,j)$. 
On the other hand, random walks are considered based on \emph{edge affinities} $a_{ij}$ which can be defined independently of the edge costs (see Section~\ref{sec:RSP} for details).
This enables the consideration of optimality and randomness of paths as being based on different grounds, which can be relevant for some applications.
The Gibbs-Boltzmann distribution is an exponential-family distribution, but can also be interpreted in terms of maximum entropy modelling~\citep{kivimaki2018distances}.

Originally the RSP framework was proposed in order to define many interpretable metrics on a network, such as distance measures between the network nodes~\citep{yen2008a-family, kivimaki2014developments} or node and edge centrality measures \citep{kivimaki2016two}.
Furthermore, the RSP framework can also be used for purposes of planning or predicting paths on networks and for modelling movement or flow patterns on networks~\citep{garcia-di2011a-continuous-state,panzacchi2016predicting}, which is the application domain considered in this paper.

One problem that has not yet been tackled in depth in the literature is the estimation of the parameters of the RSP model.
This work tackles the problem by presenting methods for computing maximum likelihood estimates (MLEs) of the parameters of the RSP model in cases where the data consists of trajectories between nodes on a network.
Such data may be generated whenever recording a movement or flow process on a network, including examples such as
\begin{itemize}
\item geolocation data of animals moving on a geographical network,
\item trip data gathered from public transportation networks, 
\item trajectory data of players in video games, or
\item browsing behavior data of web users.
\end{itemize}

The main focus in this work is on estimation of the inverse temperature parameter $\beta$ based on such data, in which case the network structure, i.e.~the edge costs and edge affinities, are considered known and fixed.
A relatively high value of $\beta$ (i.e.~low temperature) describes trajectories following optimal or near-optimal paths between the source and target nodes, whereas a low value of $\beta$ (high temperature) can describe trajectories resembling a random walk with a drift towards the target~\citep{garcia-di2011a-continuous-state}.
By setting $\beta$ to an appropriate value, the RSP model can take into account the assumption that movement or flow on a network often does not occur completely optimally nor completely randomly.
When considering movement on a network, such an assumption can be valid for various reasons.
For example, the agent (e.g.,~an animal) moving on the network might not have sufficient global information of the network or be intelligent enough in order to follow optimal paths.
Or the agent can simply have a simultaneous preference for both randomness (i.e.~exploration) and optimality (exploitation), which can help in obtaining more knowledge of the environment or in distracting a prey or an opponent trying to predict and intercept the agent.

In earlier literature, the value of $\beta$ has often been selected by some form of tuning.
However, these tuning schemes are not based on solid statistical grounds and can be very specific for the application in question.
For instance, in \citep{kivimaki2014developments}, the authors use distance measures derived from the RSP framework for clustering graph nodes, and the value of $\beta$ is fixed by maximizing the clustering performance using a part of the data.

We first derive the MLEs of $\beta$ for trajectories that have been recorded fully but also tackle the more complicated scenario where the trajectories are incomplete.
We confirm the validity of these methods with artificial examples using simulated geographic networks and an artificial community-structured network.
In addition, we test the developed MLE methods on real data of trajectories in a geographic landscape, based on GPS recordings of movements of wild reindeer.
This example is devised mostly to verify that the MLE computations can be run in practice and that the methods provide reasonable results.
Recently, the RSP framework was applied to similar data in~\citep{panzacchi2016predicting} for locating movement corridors of wild reindeer in a geographical landscape. 
In addition, it has been made familiar to the ecological community thanks to the implementation in the \texttt{gdistance} R package \citep{van2012gdistance}.
This paper is partly a continuation of the work in \citep{panzacchi2016predicting}, as we also consider trajectories of wild reindeer as a use case for the problem of fitting the model to data of movement trajectories.
However, in addition to this ecological application, the theory developed in the paper is generic so that the methods can be used for any movement or flow phenomena on networks that are suitable for the RSP model.

In addition to $\beta$, the edge costs, $c_{ij}$, and edge affinities, $a_{ij}$, can be considered as parameters of the RSP framework.
Although we do also briefly discuss the estimation of edge costs from trajectory data, in Section \ref{sec:EstimationOfEdgeCost}, the focus is mostly on the estimation of $\beta$, as commonly in RSP applications the graph structure, i.e.~the edge costs and affinities, are considered fixed and known.
Furthermore, as the purpose of this work is only in developing the RSP theory for fitting the RSP model to trajectory data, the evaluation of the MLEs in specific ecological applications and comparison with other model estimation approaches are left for future work.
We also plan to later investigate the use of the methods developed here for parameter estimation in the context of clustering and classification of graph nodes.

The paper structure is as follows: 
In the rest of this section, we recall the RSP framework and discuss its properties that are required for following the remainder of the paper.
We also discuss other work related to model estimation in similar problem settings at the end of the section.
In Section \ref{sec:CompleteTrajectories}, we derive and validate MLEs for data consisting of complete trajectories.
Section \ref{sec:IncompleteTrajectories} deals with MLEs in situations where the data consists of incomplete trajectories, i.e.,~where only a subset of the edge or node sequence of a path is observed.
In Section \ref{sec:ApplicationToWildAnimalMovement}, we use the methods derived in Section \ref{sec:IncompleteTrajectories} to fit the RSP model to data consisting of GPS trajectories of wild reindeer by estimating the inverse temperature parameter associated to the trajectories.

\subsection{Randomized shortest paths}
\label{sec:RSP}
The RSP framework has been developed during the past decade in several works \citep{yen2008a-family,saerens2009randomized,kivimaki2014developments,kivimaki2016two}, focusing mostly on the definition of distance, similarity or centrality measures on graphs or networks. 
It was initially inspired by models developed in the field of transportation science \citep{akamatsu1996cyclic}.
In this Section, we recall the definition of the RSP framework and the related results relevant for the remainder of the paper. 
Before that, however, we define the terms and notation used in the paper.

\subsubsection{Definitions and notation}
\label{sec:DefinitionsAndNotation}

\paragraph{Vector and matrix notation.} Vectors are generally denoted by lowercase boldface characters and matrices by capital boldface characters, $\mathbf{e}_{i}$ denoting the $i$-th basis vector, i.e.~1 at element $i$ and 0 elsewhere.
The length of vectors and size of matrices is determined depending on the context, if not stated explicitly.
$\mathbf{I}$ denotes the identity matrix (of appropriate size) and  $\mathbf{I}_{ij}$ denotes the matrix whose element $(i,j)$ is 1 and other elements are 0.
Otherwise, for an arbitrary matrix $\mathbf{X}$, the lowercase $x_{ij}$ and sometimes $[\mathbf{X}]_{ij}$ are used to denote the element $(i,j)$ of $\mathbf{X}$, whereas $\mathbf{X}_{ij}$ denotes the matrix containing the value $[\mathbf{X}]_{ij}$ at element $(i,j)$ and zero elsewhere.

\paragraph{Graphs and paths.}
Let $G = (V,E)$ be a directed, strongly connected graph with node set $V$ containing $n$ nodes labeled from $1$ to $n$, i.e.\ $V = \big\{1,\ldots, n \big\}$; and edge set $E$ containing $m$ edges represented as ordered pairs $(i,j)$ where $i, j \in V$ and $i \neq j$ (i.e.\ we do not consider graphs with edges from a node to itself).
For any node $i\in V$ we denote by $Succ(i)$ the set of \emph{successor} nodes of $i$, i.e.\ $Succ(i) = \big\{ j\in V \giv (i,j) \in E \big\}$.

A \emph{path} $\wp$ on $G$ is defined as a sequence of nodes $\wp = (v_{0}, \ldots, v_{L})$, where $(v_{i-1}, v_{i}) \in E$ for all $i = 1,\ldots,L$ and where $L \geq 1$ is the \emph{length} of the path (note that thus a single node does not constitute a path).
The length of an arbitrary path $\wp$ is denoted as $L(\wp)$.
The $l$-th node of path $\wp$ is denoted by $\wp(l)$, where $0 \leq l \leq L(\wp)$.
Note that the node indexing along a path starts from 0.
Likewise, the edge between the $l$-th and $(l+1)$-th node is denoted by $\wp(l,l+1)$, and the \emph{subpath} or \emph{subsequence}, from the $l_{1}$-th to the $l_{2}$-th node by $\wp(l_{1}:l_{2})$.

The focus in this work is especially on \emph{hitting paths}, i.e.\ paths, where the last node appears only once, or, formally, paths $\wp$ for which $\wp(i) \neq \wp(L(\wp))$ for all $i < L(\wp)$.
The set of all hitting paths from a starting node $s$ to a target node $t$ is denoted by $\Pset$. 
Note that with the above definition, a hitting path cannot go from a node to itself, and thus $\mathcal{P}_{tt} = \emptyset$ for all $t \in V$.
Moreover, the subset of $\Pset$ containing only paths of a fixed length $k$ is denoted by $\Pset^{(k)}$.
A pair of starting node $s$ and target node $t$ is concisely referred to as an $s$-$t$-pair and the paths in $\Pset$ as $s$-$t$-paths.
Throughout the paper, $t$ is used as the index of the \emph{absorbing} target node.
Note that the results derived in this work for hitting paths can be easily generalized to all paths, not only hitting ones.
We focus only on hitting paths for two reasons: because of brevity, and because the network measures derived from the RSP framework are more directly related to many traditional network measures when considering hitting paths than when considering all paths.

\paragraph{Edge weights, path costs and probabilities.}
In the RSP framework, each edge $(i,j)\in E$ is associated with two kinds of weights: an affinity $a_{ij}>0$ and a cost $c_{ij} > 0$.
In addition, for node pairs that are not connected by an edge, i.e.\ such $(i, j) \in V \times V$ that $(i, j) \notin E$, we define $a_{ij}=0$ and $c_{ij}=\infty$.
The edge affinities and costs define the \emph{affinity matrix} $\mathbf{A}$ and \emph{cost matrix} $\mathbf{C}$, both of size $n \times n$, whose elements $(i,j)$ are the corresponding values for the node pair $(i,j) \in V \times V$.

The edge costs $c_{ij}$ define the optimality of movement and can be considered in terms of the energy consumption, geographical distance, time duration, monetary value, or other form of expense related to the step over an edge.
The cost of a path $\wp$ is defined as the sum of the edge costs along the path:
\begin{equation}
\label{eq:PathCost}
\tc(\wp) = \sum_{l=1}^{L(\wp)} c_{\wp(l-1), \wp(l)}.
\end{equation}
Note the tilde above the $c$, which we use generally to differentiate between path-related and edge-related quantities.
For any $s$-$t$-pair, the \emph{least cost} from $s$ to $t$ means the minimum path cost over all $s$-$t$-paths.
The least cost can be considered as a directed and weighted form of the \emph{shortest path distance}.

While the edge costs define optimality of movement, the edge affinities instead determine what is meant by random movement.
Namely, the affinities determine random walks on the graph, which are generated by the transition probability matrix $\PRW$ with elements
\begin{equation}
\label{eq:ReferenceTransitionProbabilities}
\prw_{ij} 
= 
\dfrac{a_{ij}}{\sum_{k=1}^{n} a_{ik}},
\
(i,j) \in E.
\end{equation}
The random walk with the above transition probabilities corresponds to a first-order Markov chain with state set $V$.
The \emph{random walk probability distribution}, $\Prw$ over the set of hitting paths from $s$ to $t$ is determined by the product of the transition probabilities, i.e.\ for any hitting path $\wp \in \Pset$,
\begin{equation}
\label{eq:ReferencePathProbabilities}
\Prw (\wp)
=
\prod_{l=1}^{L(\wp)} \prw_{\wp(l-1),\wp(l)}.
\end{equation}
Note that this product determines the path probabilities in a well-defined manner, as it is known that for hitting paths (see, e.g.\ \citep{francoisse2017bag})
\begin{equation}
\label{eq:PrefsSumToOne}
\sum_{\wp \in \Pset} \Prw (\wp) = 1.
\end{equation}
The random walk generated by the transition probabilities $\prw_{ij}$ is sometimes referred to as the \emph{reference}, \emph{natural}, or \emph{unbiased} random walk and the superscript ``rw'' in the above quantities naturally refers to (unbiased) random walks.

The reference random walk defines the \emph{expected hitting time} and the \emph{expected hitting cost} from $s$ to $t$, as the expected path length, or path cost, respectively, over the random walk distribution $\Prw$.
The sums, from $s$ to $t$ and back from $t$ to $s$, of expected hitting times and costs define, respectively, the \emph{commute time distance} and the \emph{commute cost distance}.
It is well known that on an undirected graph (i.e.,~a graph for which $a_{ij} = a_{ji}$ and $c_{ij} = c_{ji}$ for all $(i,j) \in E$) where edge costs correspond to edge \emph{resistances} (i.e.,~$c_{ij} = r_{ij} = 1/a_{ij}$; in which case affinities correspond to conductances), both the commute time and commute cost distances are proportional to the \emph{resistance distances} between corresponding nodes.
More exactly, if we denote, for any $s$-$t$-pair on an undirected graph, by $\dist_{st}^{\mathrm{res}}$, $\dist_{st}^{\mathrm{CT}}$ and $\dist_{st}^{\mathrm{CC}}$, respectively, the resistance, commute time and commute cost distances between $s$ and $t$, then the following holds~\citep{chandra1989the-electrical,kivimaki2014developments,golnari2018random}:
\begin{equation}
\dist_{st}^{\mathrm{res}}
=
\dfrac{\dist_{st}^{\mathrm{CT}}}{\sum\limits_{(i,j)\in E}a_{ij}}
=
\dfrac{\dist_{st}^{\mathrm{CC}}}{\sum\limits_{(i,j)\in E}a_{ij} c_{ij}}.
\label{eq:dCTdCCdRes}
\end{equation}

As mentioned earlier, in Section~\ref{sec:BackgroundAndMotivation}, although the edge costs can be defined based on the edge affinities (for instance as $c_{ij} = 1/a_{ij}$, as above in the electric circuit analogue), generally in the RSP framework, the two weights can be --- and are considered in the derivations as --- independent of each other.
This assumption is convenient technically, but also makes sense due to the fact that the underlining drivers behind optimal and random movement can be very different from each other.
In the simplest examples, even when the edge costs may obtain arbitrary values, $c_{ij}\geq 0$, the random walk can be considered oblivious to the edge features by defining $a_{ij}=1$ for all $(i,j) \in E$, which results in a natural random walk with uniform transition probabilities at each node.
On the other hand, defining $c_{ij}=1$ for all $(i,j)\in E$ corresponds to considering the length of paths as their cost, in which case the random walk can nevertheless be considered with arbitrary affinities that depend on local features of edges, thus leading to a natural random walk with non-uniform transition probabilities at each node.

\subsubsection{Definition of the RSP framework}
The RSP framework can be defined and interpreted in different ways \citep{yen2008a-family,saerens2009randomized,kivimaki2014developments,kivimaki2016two,francoisse2017bag}, but here we present the most often formulated version, based on \emph{cost minimization constrained by relative entropy}.
Consider an agent moving on the graph from a starting node $s$ to a target node $t \neq s$.
Instead of moving randomly, as defined by the unbiased transition probabilities in Equation (\ref{eq:ReferenceTransitionProbabilities}), the agent aims to move in an optimal way.
However, for various possible reasons, the movement of the agent is not completely optimal.

Instead, the agent is considered to choose its path from $\Pset$ from a distribution which minimizes the expected path cost constrained to a fixed relative entropy with respect to the unbiased random walk.
Formally, we seek for the distribution that satisfies the minimization problem
\begin{equation}
\label{eq:RSPMinimization}
\mathop{\mathrm{Minimize}}_{\Prob}
\
\expC
= 
\sum_{\wp\in \Pset} 
\Prob(\wp) 
\tilde{c}(\wp)
\ \
\mathrm{s.t.}
\ \
\begin{cases}
\relent = J_{0}
\\
\sum\limits_{\wp\in \Pset} \Prob(\wp) = 1,
\end{cases}
\end{equation}
where $\relent$ is the \emph{Kullback-Leibler divergence}, or \emph{relative entropy}, with respect to the natural random walk distribution
\begin{equation}
\label{eq:KLDivergence}
\relent 
= 
\sum_{\wp \in \Pset} 
\Prob (\wp) 
\log 
\mleft( 
\Prob(\wp) 
/ 
\Prw(\wp) 
\mright),
\end{equation}
which is constrained to a fixed value $J_{0}>0$, and which determines the degree of randomness associated with the movement behavior of the agent.
A low value of $J_{0}$ constrains the distribution $\Prob$ to remain very similar to the random walk distribution $\Prw$, while a high value of $J_{0}$ allows the distribution to focus more on optimal, low-cost paths.

The solution of (\ref{eq:RSPMinimization}), which can be obtained with the Lagrangian method (see, e.g.\ \citep{yen2008a-family}), is the Gibbs-Boltzmann distribution (which we also refer to as the \emph{RSP distribution}) over the set $\Pset$:
\begin{equation}
\label{eq:GibbsBoltzmann}
\Prob(\wp) 
= 
\dfrac{\Prw(\wp) \exp(-\beta\tc(\wp))}{\sum\limits_{\wp\,' \in \Pset} \Prw(\wp\,') \exp(-\beta\tc(\wp\,'))},
\end{equation}
where $\beta = 1/T > 0$ is the \emph{inverse temperature parameter}, resulting from introducing $T$ as the Lagrangian multiplier of the relative entropy constraint.
The parameter $\beta$ is thus related to the relative entropy value $J_{0}$.
Namely, for low values of $\beta$ (corresponding to low values of $J_{0}$), i.e.\ when $\beta \rightarrow 0^{+}$, the distribution converges to the random walk distribution.
For high values of $\beta$ (high values of $J_{0}$), i.e.\ when $\beta \rightarrow \infty$, the distribution focuses more and more on the low-cost paths.

Among the other interpretations of the RSP framework, the RSP distribution could be obtained from an inverse point-of-view as well, namely by considering minimization of relative entropy with a fixed expected cost.
That way the RSP framework can be interpreted as a maximum entropy model~\cite{kapur1992entropy}, as minimizing relative entropy with respect to the random walk distribution is conceptually close, though not equivalent, to maximizing the Shannon entropy.

Normally in the RSP model, for convenience, the user is required to set the value of $\beta$, instead of $J_{0}$, to determine the degree of randomness associated to the distribution.
Note that there is no analytical expression for computing the value of $\beta$ that would result in the RSP distribution with a given relative entropy $J_{0}$.
Instead, for obtaining the RSP distribution with a particular relative entropy, the corresponding value of $\beta$ can be obtained with a bisection search.
Also worth noting is that Equation~(\ref{eq:GibbsBoltzmann}) is the solution of the problem in Equation~(\ref{eq:RSPMinimization}) only if $J_0 < J_{\mathrm{max}}$, where $J_{\mathrm{max}}$ is the relative entropy at the limit $\beta \longrightarrow \infty$, where the RSP distribution is focused only on the shortest paths and zero for other paths.

Let us denote the numerator in (\ref{eq:GibbsBoltzmann}), i.e.\ the \emph{likelihood of path $\wp$}, as
\begin{equation}
\label{eq:PathLikelihood}
\widetilde{w}(\wp) = \Prw(\wp)\exp(-\beta \tc(\wp)).
\end{equation}
In fact, we make this definition for all paths between $s$ and $t$, not only the hitting ones.
The denominator in Equation (\ref{eq:GibbsBoltzmann}) is the \emph{partition function} of hitting $s$-$t$-paths, which accumulates the overall likelihood of hitting $s$-$t$-paths, and which is denoted by
\begin{equation}
\label{eq:PF}
\PF = \Psum \Prw(\wp) \exp(-\beta\tc(\wp))
=
\Psum \tw(\wp).
\end{equation}
Finally, by defining the \emph{edge likelihoods} by 
\begin{equation}
\label{eq:wij}
w_{ij} = \prw_{ij}\exp(-\beta c_{ij}),
\end{equation}
the path likelihood is, in fact, the product of the edge likelihoods along the path:
\begin{align}
\label{eq:PathLikelihoodFromEdgeLikelihoods}
\widetilde{w}(\wp) 
= \prod_{l=1}^{L(\wp)} w_{\wp(l-1),\wp(l)}.
\end{align}

\subsubsection{Computation of main quantities}
\label{sec:ComputationOfMainQuantities}
Here we gather results from earlier literature for computing quantities from the RSP distribution that are relevant for the current work.
We first present the computation of the path likelihoods $\widetilde{w}$ and the partition function, $\PF$, based on matrix computations, and then show how, using the partition function, we may compute the \emph{expected cost} of paths from $s$ to $t$, as well as the \emph{expected number of traversals over edges} and the \emph{expected number of visits to nodes}, over the RSP distribution.
All of these quantities and their computation appear throughout the derivations related to the maximum likelihood estimation in later sections.

The edge likelihoods $w_{ij}$, from Equation~(\ref{eq:wij}), define the \emph{likelihood matrix} $\mathbf{W}$, which can be expressed as
\begin{equation}
\label{eq:Wmatrix}
\mathbf{W} = \PRW \circ \exp(-\beta \mathbf{C}),
\end{equation}
where $\PRW$ and $\mathbf{C}$ are the matrices containing the reference transition probabilities and edge costs, respectively, and $\circ$ is the element-wise, i.e.~Hadamard, product, and the exponential is taken element-wise as well.
The likelihood matrix $\mathbf{W}$ is substochastic, i.e.~its row sums are all less than unity, $\sum_{j} w_{ij} < 1$.
This can be interpreted as $\mathbf{W}$ defining a \emph{killed random walk} (also sometimes called the \emph{evaporating random walk}), where the residue probability at each node $i$, $1-\sum_{j} w_{ij}$, corresponds to the probability of transition from $i$ to an imaginary, absorbing ``cemetery node'', instead of continuing the walk to a neighbouring node.

We present two ways of computing the partition function of Equation~(\ref{eq:PF}).
The first way is to define matrix 
\begin{equation}
\Wt
=
\mathbf{W}
-
\mathbf{I}_{tt}
\mathbf{W}
\label{eq:Wt}
\end{equation}
as the matrix $\mathbf{W}$ with row $t$ set to zero.
This is equivalent to considering a deletion of all edges leaving node $t$, which makes $t$ an \emph{absorbing} node (or state; analogous to absorbing Markov chains~\citep{grinstead1997introduction}).

The overall likelihood of hitting paths from $s$ to $t$ of given length $k=1,2,\ldots$, is given by elements of the powers of $\Wt$:
\begin{equation}
\label{eq:WPowers}
\sum_{\wp \in \Pset^{(k)}} 
\widetilde{w}(\wp)
= 
\left[
\Wt^k
\right]_{st}.
\end{equation}
As a result, the partition function, $\PF$, defined earlier in Equation~(\ref{eq:PF}), can be computed by summing over all path lengths $k$:
\begin{equation}
\PF
=
\left[
\Wt
+
\Wt^2
+
\Wt^3
+
\cdots
\right]_{st}
=
\left[
(
\mathbf{I}-
\Wt
)^{-1}
- \mathbf{I}
\right]_{st}.
\label{eq:PFFromWt}
\end{equation}

However, computing quantities with the above approach can be costly, when considering different target nodes, $t$, as the matrix inverse appearing in Equation~(\ref{eq:PFFromWt}) has to be computed separately for each $t$.
A second way of computing the partition function $\PF$ is to first compute the \emph{fundamental matrix of all paths}, given by
\begin{equation}
\label{eq:Zmatrix}
\mathbf{Z} 
= \mathbf{I} + \mathbf{W} + \mathbf{W}^2 + \mathbf{W}^3 + \cdots
= \mathbf{(I-W)}^{-1},
\end{equation}
whose elements $z_{st}$ quantify the expected number of visits to node $t$ before being killed (i.e.\ transitioning to the cemetery node) during a killed random walk starting from node $s$ based the substochastic transition matrix $\mathbf{W}$.
Then, as was shown in \citep[Appendix B]{francoisse2017bag} and in \citep{kivimaki2014developments}, the partition functions $\PF$, from Equation~(\ref{eq:PF}), for any $s$ and $t$ such that $s \neq t$, can be computed based on elements of matrix $\mathbf{Z}$ of Equation~(\ref{eq:Zmatrix}) as
\begin{equation}
\label{eq:PFFromZ}
\PF = z_{st}/z_{tt}.
\end{equation}
Note, that, for any $s$-$t$-pair, the italic $z_{st}$ always refers to element $(s,t)$ of matrix $\mathbf{Z}$, whereas the calligraphic $\PF$ denotes the partition function, defined in Equation~(\ref{eq:PF}).
Based on this result, as discussed in \citep{francoisse2017bag} and \citep{kivimaki2014developments}, for any $s$ and $t$ such that $s \neq t$, the partition function of hitting paths $\PF$ can be shown to quantify the probability that a walker starting from node $s$ and moving according to the substochastic transition matrix $\mathbf{W}$ survives to node $t$ before being killed.

Finally, the matrix whose element $(s,t)$ contains the partition function $\PF$ from Equation~(\ref{eq:PF}) for all $(s,t) \in V \times V$ can be expressed, based on Equation~(\ref{eq:PFFromZ}), as
\begin{equation}
\label{eq:Zhmatrix}
\mathbf{Z}_{\mathrm{h}} 
= 
\mathbf{Z}
\mathbf{D}_{\mathbf{Z}}^{-1}
-
\mathbf{I},
\end{equation}
where $\mathbf{D}_{\mathbf{Z}}$ is the $(n \times n)$ diagonal matrix of the diagonal elements of $\mathbf{Z}$.
Although this expression is more convenient for computing the partition functions between multiple (or all) $s$-$t$-pairs at once, in this paper we, however, rely more on the form in Equation~(\ref{eq:PFFromWt}), considering one target node $t$ at a time.

Note that the nonnegativity, substochasticity and irreducibility (as $G$ is strongly connected) of $\mathbf{W}$ imply that the spectral radius of $\mathbf{W}$ is less than unity, $\rho(\mathbf{W}) < 1$.
Also, although $\Wt$ is not irreducible, as the graph is not strongly connected after the removal of the edges leaving $t$, we nevertheless have $\rho(\Wt) < 1$, as setting a row to zero in a matrix cannot increase its spectral radius.
This ensures, based on Perron-Frobenius theory, that the matrix series in Equations~(\ref{eq:PFFromWt}) and (\ref{eq:Zmatrix}) converge and can be computed using the presented matrix inverses~\citep{meyer2000matrix}.

The partition function $\PF$ can be manipulated in order to derive the computation of various meaningful quantities related to the RSP framework.
In particular, the \emph{expected number of traversals over an edge $(i,j)$}, when moving according to the RSP distribution over hitting paths from $s$ to $t$, is given by (see, e.g.\ \citep{kivimaki2016two})
\begin{equation}
\label{eq:ExpVij}
\expVij(s,t)
=
\sum_{\wp\in \Pset} 
\Prob(\wp) 
\Vij(\wp)
=
-\dfrac{1}{\beta}
\dfrac{\partial \log \PF}{\partial c_{ij}}
,
\end{equation}
where $\Vij(\wp)$ denotes the number of times edge $(i,j)$ appears on path $\wp$.
This derives from the fact that 
\begin{equation}
\label{eq:VisitsOverP}
\frac{\partial \tc (\wp)}{\partial c_{ij}}  = \Vij(\wp).
\end{equation}

Similarly, the \emph{expected cost} of moving from $s$ to $t$ when moving according to the RSP distribution over hitting paths from $s$ to $t$ is (see e.g.\ \citep{kivimaki2014developments})
\begin{equation}
\label{eq:ExpectedCost}
\expCst 
= 
\sum_{\wp\in \Pset} 
\Prob(\wp) 
c(\wp)
=
-\dfrac{\partial \log \PF}{\partial \beta}
.
\end{equation}
By altering the temperature, the expected cost over the RSP distribution interpolates between the least cost (when $\beta \rightarrow \infty$) and the expected hitting cost (when $\beta \rightarrow 0^{+}$) from $s$ to $t$.
Accordingly, the symmetrized version $\expCst + \expCts$ interpolates between the least cost distance (multiplied by $2$), and the commute cost distance.

Thanks to the above derivative expressions (\ref{eq:ExpVij}) and (\ref{eq:ExpectedCost}), the above quantities can be expressed in terms of elements of matrices $\mathbf{W}$ (from Equation~(\ref{eq:Wmatrix})) and $\mathbf{Z}$ (from Equation~(\ref{eq:Zmatrix})) as (again, see \citep{kivimaki2014developments, kivimaki2016two})
\begin{equation}
\label{eq:ExpVComputation}
\expVij(s,t)
=
\mleft(
\dfrac{z_{si}}{z_{st}}
-
\dfrac{z_{ti}}{z_{tt}}
\mright)
w_{ij}
z_{jt}
\end{equation}
and
\begin{equation}
\label{eq:ExpCComputation}
\expCst 
=
\sum_{(i,j)\in E}
\expVij(s,t)
c_{ij}
=
\sum_{(i,j)\in E}
\mleft(
\dfrac{z_{si}}{z_{st}}
-
\dfrac{z_{ti}}{z_{tt}}
\mright)
w_{ij}
c_{ij}
z_{jt}.
\end{equation}

From (\ref{eq:ExpVComputation}), the \emph{expected number of visits to node $i$} with respect to the RSP distribution over hitting paths from $s$ to $t$, can be computed as
\begin{equation}
\label{eq:ExpVNode}
\expVi(s,t)
=
\sum_{j\in Succ(i)}
\expVij(s,t)
=
\mleft(
\dfrac{z_{si}}{z_{st}}
-
\dfrac{z_{ti}}{z_{tt}}
\mright)
z_{it}.
\end{equation}
The quantity $\expVi = \sum_{s \in V} \sum_{t \in V} \expVi(s,t)$ was coined in~\citep{kivimaki2016two} as the \emph{simple RSP betweenness centrality} of node $i$, which interpolates between the \emph{shortest path likelihood betweenness} (when $\beta \rightarrow \infty$), which is strongly related to the standard shortest path betweenness centrality~\citep{freeman1978centrality}, and the stationary distribution of the unbiased random walk (when $\beta \rightarrow 0^{+}$) on the graph.

The RSP distribution over paths from $s$ to $t$ can also be interpreted as defining a \emph{biased random walk}, with new transition probabilities containing a drift towards $t$.
The \emph{biased transition probabilities} towards $t$ can be obtained by using (\ref{eq:ExpVComputation}) and (\ref{eq:ExpVNode}) as
\begin{equation}
\label{eq:Pbias}
p_{ij}^{(t)} 
= 
\dfrac{\expVij}{\expVi}
=
\dfrac{w_{ij}z_{jt}}{z_{it}},
\end{equation}
for all $i \neq t$.
As we consider hitting paths, $\bar{\V}_{t} = 0$, and the biased transition probabilities are separately defined as zero for the target node, i.e.~$p_{tj}^{(t)} = 0$ for all $j$.
As can be seen from Equation (\ref{eq:Pbias}), the biased transition probabilities are independent of $s$, i.e.~$p_{ij}^{(t)}$ is the same for any edge $(i,j)$ for all starting nodes $s$.
The biased transition probabilities can be used for generating individual paths over a graph according to the RSP distribution.

\subsection{Related work}
\label{sec:RelatedWork}
The RSP framework was originally defined in order to develop distance and centrality measures on graphs for graph-based machine learning purposes~\citep{saerens2009randomized,yen2008a-family,kivimaki2014developments}.
Variants of the RSP framework have been developed for different use cases, including the sum-over-paths \citep{mantrach2010sum} and bag-of-paths \citep{francoisse2017bag} frameworks.
Parameter tuning in applications of these frameworks (e.g.~in clustering in \citep{kivimaki2014developments} and semi-supervised classification in \citep{lebichot2014semisupervised}) has usually been dealt with by using a held-out tuning data set, and by searching for the parameter value giving best performance on this data.
In this paper we approach the parameter estimation problem from a more fundamental view.
Note, however, that the setting is quite different from the previous applications of the RSP framework, and we leave for future work the investigation of using the results derived here for parameter tuning in other application areas of RSPs.

The RSP framework shares similarities with the \emph{logit assignment model} and its different variants proposed in transportation science \citep{dial1971probabilistic,ben1985discrete,ben1999discrete,prashker2004route}, and was originally inspired by such models~\citep{akamatsu1996cyclic}.
Logit assignment models are sometimes criticised for a couple of reasons. 
First, they can become computationally untractable when considering cycles on the network because of infinitely cumulating costs~\citep{oyama2017discounted}.
This problem has been alleviated, for instance, by restricting the path set to ``efficient'' paths~\citep{dial1971probabilistic}, by constraining the path lengths~\citep{oyama2019prism}, or by introducing a \emph{discount factor}~\citep{oyama2017discounted} that makes the model interpolate between unbiased and optimal behaviour similarly (but not equivalently) to the inverse temperature parameter $\beta$ in the RSP model.
A second criticism about logit assignment is that it has the \emph{independence from irrelevant alternatives (IIA)} property~\citep{ben1985discrete}, essentially meaning that an agent leaving from $s$ to $t$ according to the logit probabilities will more likely select an edge  that leads to more alternative $s$-$t$-paths, even if those paths are almost identical (i.e.\ irrelevant).
This issue has been addressed, for instance, with different \emph{nested} models~\citep{ben1973structure,mai2015nested}, where the set of paths (or \emph{choice set}) is partitioned into nests containing irrelevant alternatives.

The logit assignment models originate from \emph{random utility theory}, where edge costs (or utilities) are considered as a sum of a deterministic part and a random (``error'' or ``noise'') part, which leads to a logit distribution over paths similar to the Gibbs-Boltzmann distribution of the RSP framework~(see, e.g., \citep{ben1985discrete}) when minimizing expected cost.
In such a formulation, the expected cost actually corresponds to the concept of \emph{free energy}, or \emph{potential} in the RSP formalism~\citep{kivimaki2014developments,francoisse2017bag}. 
However, normally the RSP free energy is considered as the \emph{relative} free energy by regularizing costs in terms of relative entropy (instead of the Shannon entropy) with respect to the unbiased random walk, (see Section~\ref{sec:BackgroundAndMotivation}).
Thanks to this difference, the RSP model avoids the issue of intractability of computation that the logit assignment model faces with cycles on the network.
In more detail, the regularization based on relative entropy results in the substochasiticity of the likelihood matrix $\mathbf{W}$ defined in Equation~(\ref{eq:Wmatrix}), which ensures that the RSP model can be computed for any positive value of the inverse temperature parameter $\beta$, without having to restrict the set of paths.
The RSP model also alleviates the IIA issue in some cases, as the addition of an alternative route causes a change in the unbiased random walk, and thus the RSP probabilities.

The standard derivation of the RSP model considers cost-minimization subject to constrained relative entropy.
The relative entropy can be interpreted information-theoretically as a fixed degree of informedness, or knowledge that the random walker has of the environment (compared to an unbiased random walker); in other words, the walker aims at minimizing travel cost constrained on its knowledge of the network.
Such an information-theoretic interpretation can be expressed for some logit assignment models, as e.g.\ the \emph{rational inattention model} in \cite{matvejka2015rational}.
On the other hand, the RSP model can also be derived in a similar fashion to the derivation of the logit assignment model, based on augmented edge costs, and minimization of expected augmented path costs.
We leave these developments for future work, but mention them here to increase the motivation of using RSPs in modelling trajectory data, in addition to the discussion on the topic already in Section~\ref{sec:BackgroundAndMotivation}.

The literature behind logit assignment models is extensive, and also includes various studies for parameter estimation, also with maximum likelihood methods. 
The early works \citep{robillard1974calibration,fisk1977note} derive maximum likelihood methods for a parameter of the logit assignment model corresponding to the inverse temperature in our formulation.
Their results are similar to the results derived in our work for complete trajectories, although the works only consider networks without cycles.
Also, the \emph{dispersion-constrained model} in~\citep{anas1988statistical} is based on a maximum-entropy approach similar to the RSP framework, and involves the estimation of a temperature parameter.
Other works, such as \citep{fosgerau2013link,mai2015nested,mai2016method,oyama2017discounted,oyama2018link,oyama2019prism} focus, in various traffic-related settings, on the problem of estimating parameters on edges or parameters associated to features on edges, which in some cases corresponds to the problem of estimating the inverse temperature and in some case to the problem of estimating edge costs in the RSP framework (the latter of which is addressed only briefly in the current work, in Section \ref{sec:EstimationOfEdgeCost}).
The above cited works in logit assignment and related models is in many respects very close and similar to the RSP model.
However, one key difference is that, as already discussed above, the RSP model considered in this paper, by maximizing relative entropy, generalizes a pure, unbiased random walk, whereas the logit assignment models can be seen as maximizing the Shannon entropy over paths.
In addition, to the best of our knowledge, a similar definition and handling of data consisting of incomplete trajectories, as presented here, does not seem to appear in the logit assignment literature.

The problem of parameter estimation in the RSP framework has similarities with the more general problem of estimation of Markov chain transition probabilities from trajectory data (see e.g.~\citep{craig2002estimation, metzner2007generator, shelton2014tutorial}).
Indeed, when assuming the RSP model, a given value of $\beta$ defines the biased transition probabilities towards the target node, according to Equation~(\ref{eq:Pbias}), and thus estimating $\beta$ implies estimating the transition probabilities.
However, the general estimation of Markov chains, studied in the works cited above, does not involve assumptions on the distribution of paths on the graph, and also does not consider affinities, costs or other weights related to the transitions.
Moreover, Markov chain estimation has been mostly studied for continuous-time Markov jump processes, for the estimation of the \emph{generator matrix} of the process.
Such techniques have been developed e.g.~for healthcare research, for studying disease progression and treatment effects, where the temporal aspect is of course vital~\citep{craig2002estimation}.
Instead, the RSP framework is only analogous to a discrete time Markov chain, although extending the idea of RSPs to consider continuous time is a planned topic for future research.

Another problem related to the current work, especially the problem of model inference with incomplete trajectory data, is the inference of \emph{aggregate Markov chains}, which has been applied to ion channel modelling \citep{qin1997maximum}.
There the state space is divided into aggregates and only the aggregate that the system is in, instead of the actual state, is observed.
Aggregated Markov chains have also been only considered in the context of continuous-time chains.
Aggregated Markov chains bare similarities with the study of \emph{stochastic complementation} in discrete-time Markov chains~\citep{meyer1989stochastic}.

Besides the temporal aspect, there are also other differences between the general estimation of Markov chains and the estimation of RSP parameters.
Namely, the methods for Markov chain estimation normally require data containing observations at each state of the chain, as the transition probabilities for unobserved states cannot be estimated \citep{bladt2005statistical}.
The estimation of the inverse temperature parameter in the RSP model, however, does not require this, and it can be estimated even for a sparse set of incomplete trajectories, as derived and verified with experiments in Section~\ref{sec:IncompleteTrajectories}.

\section{Maximum likelihood estimation with complete trajectories}
\label{sec:CompleteTrajectories}
The rest of the paper is dedicated to deriving and verifying methods for computing the likelihood of the RSP model parameters when modelling trajectory data, and computing the maximum likelihood estimates (MLEs).
The focus is mostly on the likelihood of the inverse temperature parameter $\beta$.
For estimating $\beta$, the \emph{graph structure}, i.e.~the edge affinities (or the reference transition probabilities) and edge costs, are always assumed to be known completely and exactly.
The MLE of $\beta$ is denoted in general as $\betaMLE$.
However, in the expressions for the MLEs, $\beta$ does not normally appear explicitly, but the expressions involve other RSP quantities, which can be computed with the methods detailed in Section~\ref{sec:ComputationOfMainQuantities}.
Accordingly, $\betaMLE$ can then be found by performing a search that satisfies the expressed MLE criterion, or by other, more direct optimization means.

This section deals with data sets of \emph{complete trajectories} on a network, i.e.~where each node (and thus each edge) of the observed path has been recorded.
We derive the results by first considering the simple case where all trajectories are observed to go from one starting node $s$ to one destination node $t$, and then extend to trajectories between several $s$-$t$-pairs.
The methods are then validated with artificial data.
We also present briefly in this section a method for estimating the edge costs $c_{ij}$ from complete trajectories, when assuming that only the affinities or reference transition probabilities are known, but otherwise the focus of the paper is on estimating $\beta$.
Later, in Section \ref{sec:IncompleteTrajectories}, we tackle the problem of computing $\betaMLE$ when the data consists of incomplete trajectories, where only a part of the edges or nodes visited along each trajectory is observed.

\subsection{Single source and target}
Let $\Omega_{st}$ be a data set containing $K$ fully observed hitting paths going from $s$ to $t$.
\begin{theorem}
\label{thm:CompleteTrajs_st}
Assuming independence between the trajectories, the maximum likelihood estimate of the inverse temperature, $\hat{\beta}_{\mathrm{MLE}}$, given the data set $\Omega_{st}$, as described above, is the value of $\beta$ for which the expected cost of hitting $s$-$t$-paths over the RSP distribution, $\expCst$, from Equations~(\ref{eq:ExpectedCost}) and (\ref{eq:ExpCComputation}), satisfies,
\begin{equation}
\expCst 
= 
\dfrac{1}{K}
\sum_{\wp \in \Omega_{st}} 
\tc(\wp),
\end{equation}
\end{theorem}

\begin{proof}
Based on the assumption of independence between the trajectories, the likelihood is simply the product of the individual path probabilities from Equation~(\ref{eq:GibbsBoltzmann}):
\begin{equation}
\begin{aligned}
\label{eq:FullSinglePathLikelihood}
\mathcal{L}
(\beta \giv \Omega_{st})
= 
\prod_{\wp \in \Omega_{st}} 
\Prob(\wp)
= 
\prod_{\wp \in \Omega_{st}} 
\dfrac{\Prw(\wp) \exp(-\beta \tc(\wp))}{\PF}
\end{aligned}
\end{equation}
and the log-likelihood is
\begin{equation}
\label{eq:FullSinglePathLogLikelihood}
\logL
= 
\sum_{\wp \in \Omega_{st}} 
\mleft(
\log \Prw(\wp) - \beta \tc(\wp) 
\mright) 
- K \log \PF.
\end{equation}
Taking the derivative of the log-likelihood with respect to $\beta$, and setting it to zero gives us the necessary optimality condition for the MLE:
\begin{equation}
\label{eq:FullSinglePathLikelihoodDerivative}
\dfrac{\partial \logL}{\partial \beta}
= -\sum_{\wp \in \Omega} \tc(\wp) - K \dfrac{\partial \log \PF}{\partial \beta}
= 0
\ 
\Longleftrightarrow
\ 
\expCst 
= 
\dfrac{1}{K}
\sum_{\wp \in \Omega} 
\tc(\wp)
\end{equation}
where we used Equation (\ref{eq:ExpectedCost}).
\end{proof}

Theorem~\ref{thm:CompleteTrajs_st} simply statest that the likelihood is maximized by the value $\betaMLE$ for which the expected cost over the RSP distribution equals the average cost of the observed trajectories.
A similar result is common in a large class of maximum entropy estimation procedures~\citep{kapur1992entropy}, whereas it is presented here for the first time in the literature related to RSPs.
As discussed above, Equation~(\ref{eq:FullSinglePathLikelihoodDerivative}) cannot be solved explicitly for $\beta$. 
Instead, the value $\betaMLE$ has to be estimated numerically by finding the root of the equation.

\subsection{Multiple sources and targets}
\label{sec:CompleteTrajsMultipleST}

Let then $\Omega$ denote a data set of trajectories between different $s$-$t$-pairs, $\Omega_{st} \subset \Omega$ the set of trajectories in the data that go from a particular $s$ to a particular $t$, and $K_{st} = | \Omega_{st} |$ the number of trajectories from $s$ to $t$ (with $K_{st}=0$ if none of the trajectories go from $s$ to $t$).

\begin{theorem}
\label{thm:CompleteTrajs}
Given data set $\Omega$, and assuming, again, independence between the trajectories, $\betaMLE$ is the value for which the RSP expected costs, $\expCst$, for all $s$ and $t$ satisfy
\begin{equation}
\sum_{s,t \in V} 
K_{st} 
\mleft< \tc \mright>_{st} 
= 
\sum_{s,t \in V} 
\sum_{\wp_{st} \in \Omega_{st}} 
\tc(\wp_{st})
\end{equation}
\end{theorem}

\begin{proof}
The likelihood, given $\Omega$, is simply the product of the likelihoods for each $s$-$t$-pair and, accordingly, the log-likelihood can be written as
\begin{equation}
\label{eq:MultipleFullPathsLogLikelihood}
\logL 
(\beta \giv \Omega)
= \sum_{s,t \in V} \! \mleft( - K_{st} \log \PF + \sum_{\wp_{st} \in \Omega_{st}} \mleft( \log \Prw(\wp_{st}) - \beta \tc(\wp_{st}) \mright) \mright).
\end{equation}
Again, as in the proof of Theorem~(\ref{thm:CompleteTrajs_st}), setting the derivative to zero, we see that the MLE of $\beta$ should satisfy
\begin{align}
\sum_{s,t \in V} 
K_{st} 
\mleft< \tc \mright>_{st} 
= 
\sum_{s,t \in V} 
\sum_{\wp_{st} \in \Omega_{st}} 
\tc(\wp_{st}).
\label{eq:MultipleFullPathsLikelihoodDerivative}
\end{align}
\end{proof}

Thus, as in Theorem~\ref{thm:CompleteTrajs_st}, as well as often in maximum entropy maximum entropy estimation methods~\citep{kapur1992entropy}, the most likely value of $\beta$ corresponds to the one for which the empirical average agrees with the expected value given by the model.
Likewise, as in Section~\ref{sec:CompleteTrajectories}, Equation~(\ref{eq:MultipleFullPathsLikelihoodDerivative}) cannot be solved explicitly for $\beta$, but $\betaMLE$ has to be determined numerically.

\subsection{Validation of $\betaMLE$ with complete trajectories}
\label{sec:ValidationOfFullPathMLE}

\begin{figure}
	\centering
	\subcaptionbox{\label{fig:GaussianLandscape}}
	{\includegraphics[width=0.48\textwidth, trim=2cm 1.2cm 1.2cm 0.8cm, clip=true]{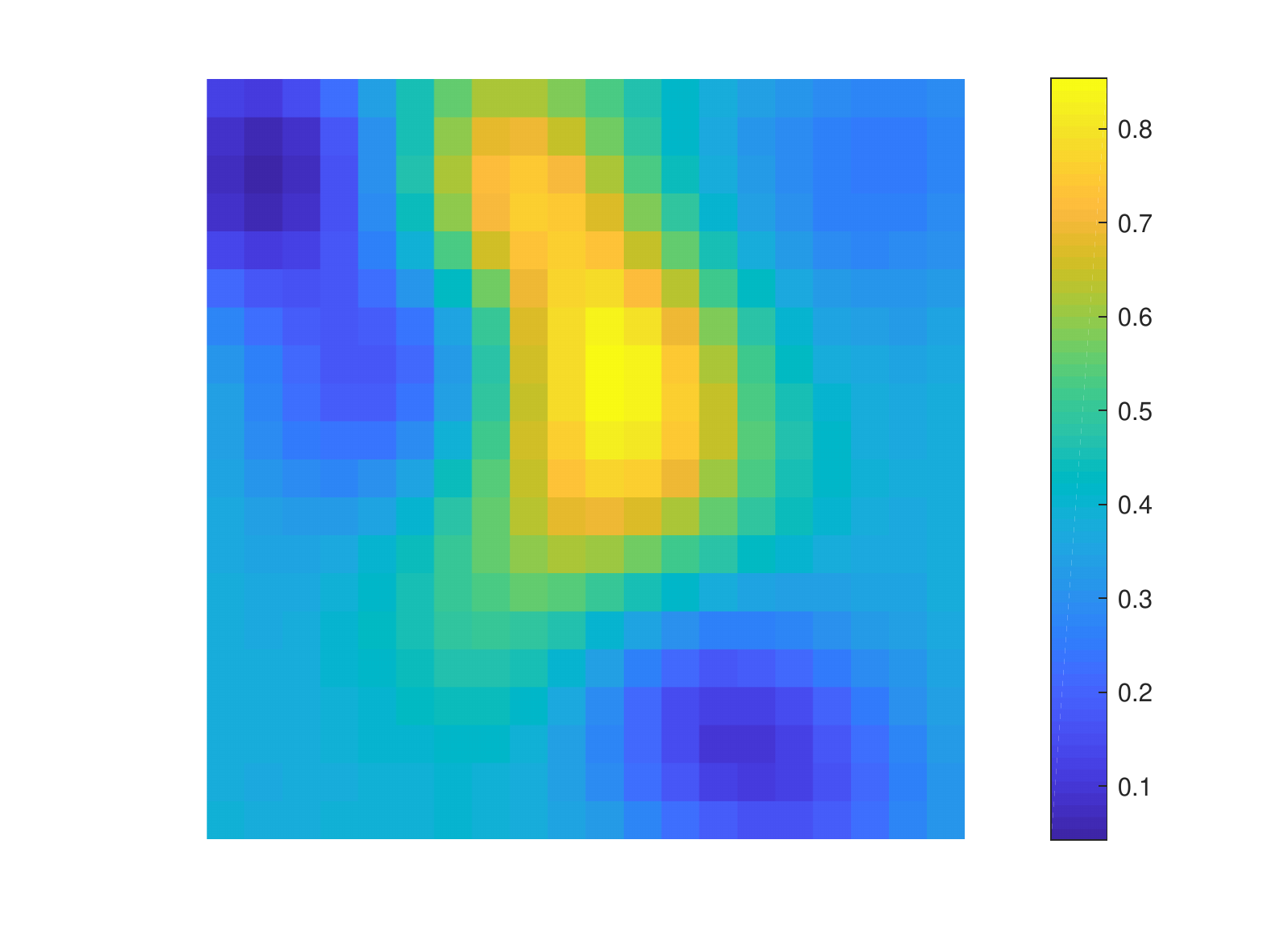}}	
	\subcaptionbox{\label{fig:LFR_graph}}
	{\includegraphics[width=0.48\textwidth, trim=4cm 4cm 4cm 4cm, clip=true]{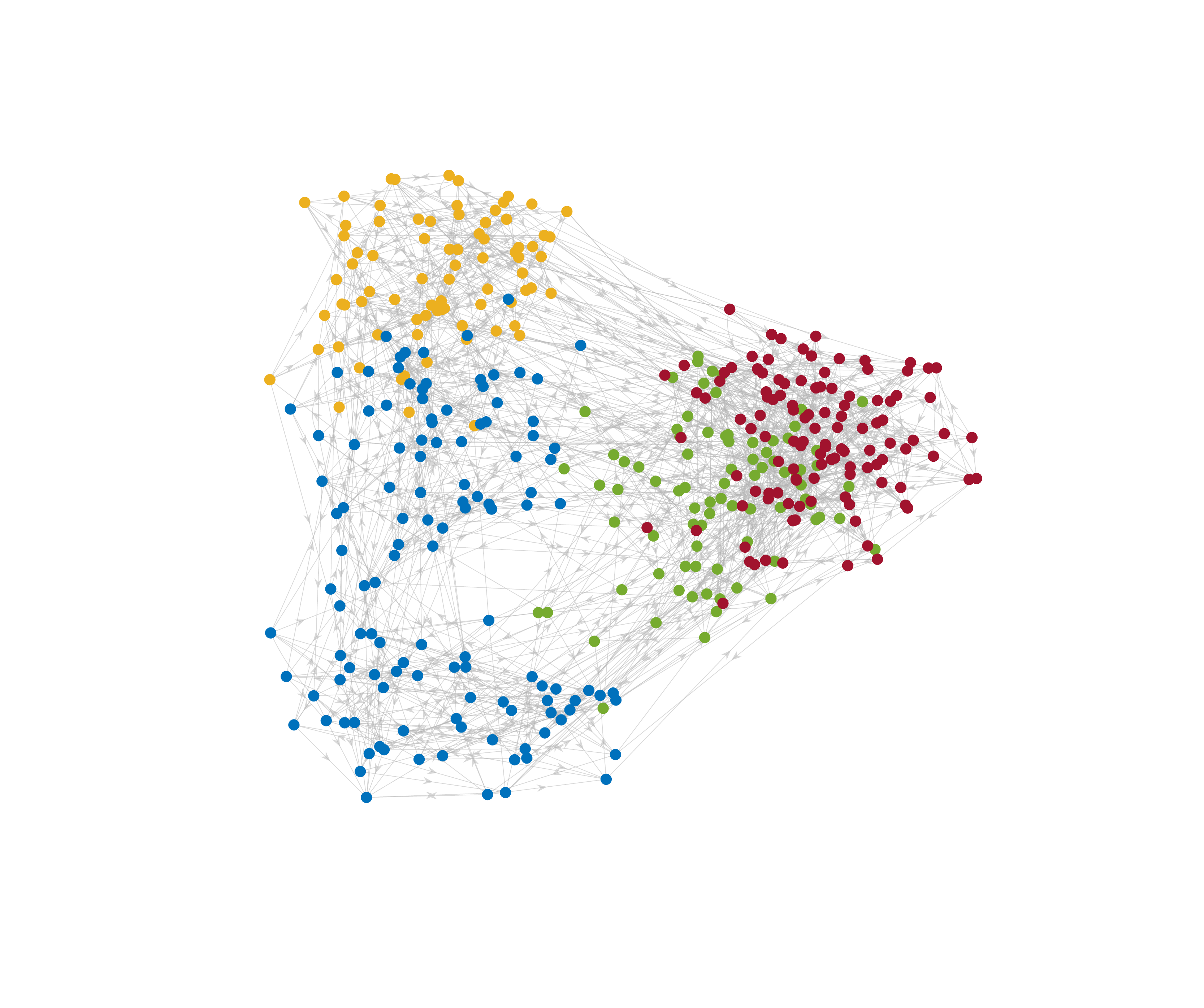}}
	\caption{Two of the three graphs used for validating the MLE methods. (A): The generated $20 \times 20$ Gaussian landscape. Edge costs are determined by costs on pixels, shown by the heatmap, so that for pixel $j$, $c_{ij} = c_{j}$ for all $i$ s.t.~$(i,j) \in E$. (B): The graph generated with the LFR algorithm. The node colors depict the nodes' communities.
	}
	\label{fig:ExperimentGraphs}
\end{figure}

The above results and their accuracy were evaluated by generating trajectories on three artificial graphs.
The first graph is a simple $20 \times 20$ grid with uniform edge costs.
The second is a $20 \times 20$ Gaussian landscape depicted as a heatmap in Figure~\ref{fig:GaussianLandscape}.
In that graph, each node (corresponding to a pixel in the heatmap) is assigned a cost based on a mixture of Gaussian distributions on a plane overlaid on the grid. 
In more detail, the mixture consists of five low-cost and five high-cost Gaussian patches, which cause decrease and increase of the costs from a base value of $0.5$.
The edge costs on the graph are then determined by the cost assigned to the ending node of the edge.
In both the uniform grid and the Gaussian landscape, each node is connected to its adjacent and diagonal neighbors and the cost of the diagonal connections is multiplied by $\sqrt{2}$.
Affinities are fixed as reciprocals of costs, i.e.~$a_{ij} = 1/c_{ij}$, and reference transition probabilities (as explained earlier) as the normalized affinities, $\prw_{ij} = a_{ij} / \sum_{k} a_{ik}$.

The third graph used for validating our results is a weighted directed graph generated with the Lanchichinetti-Fortunato-Radicchi (LFR) algorithm~\cite{lancichin2009benchmarks} designed for generating artificial complex networks with a community structure.
We simply generated one LFR graph with 400 nodes forming 5 communities, with the mixing parameters for nodes and weights set to $0.2$, the degree sequence exponent to $4$ and the average degree of nodes set to $4$.
A visualization of the resulting LFR graph used in the experiments is depicted in Figure~\ref{fig:LFR_graph}, where the node colors represent the community of each node.
The weighted-network version of the LFR algorithm generates edge weights $w_{ij}$ that reflect the strength of connection between nodes respecting the community structure.
In order to study the effect of the possible independence between edge affinities and costs on the estimation results, we defined the affinities on LFR graph simply as $a_{ij} = 1$ for all $(i,j) \in E$, and used the weights $w_{ij}$ generated by the LFR algorithm for defining costs by taking the reciprocal, i.e., $c_{ij} = 1/w_{ij}$.
Thus, considering this graph in the validation experiments, in addition to the two grid graphs introduced above, ensures that the methods developed in this work also work properly with other network structures as well as with independent edge affinities and costs.

Using a set of values for $\beta$, 200 paths were generated on both landscapes based on each studied value of $\beta$.
Each path was generated by first drawing an $s$-$t$-pair uniformly randomly on the grid, however only accepting node pairs that were at least 3 steps apart on the grid. 
A path between each $s$-$t$-pair was then generated by using the biased transition probabilities from Equation~(\ref{eq:Pbias}).
Then, for each such set of paths $\Omega$, the value $\betaMLE(\Omega)$ was inferred by finding, by a simple line search, the value of $\beta$ that satisfied Equation (\ref{eq:MultipleFullPathsLikelihoodDerivative}).
The above procedure was repeated 10 times for each studied value of $\beta$ and the mean and standard deviation of the MLEs was recorded.
The results in Table~\ref{tab:MultipleFullPaths} show that the MLEs, on average, indeed closely match the true values of $\beta$ used for generating the paths.

\begin{table}
\begin{center}
{\footnotesize 
\begin{tabular}{r|c c c}
$\beta$
& $\hat{\beta}_{\mathrm{MLE}}$, uniform grid
& $\hat{\beta}_{\mathrm{MLE}}$, simulated landscape
& $\hat{\beta}_{\mathrm{MLE}}$, LFR graph
\\
\hline
$   0.001$ & $  0.00096 \pm 0.00020$ & $  0.00111 \pm 0.00024$ & $  0.00113 \pm 0.00032 $\\
$   0.005$ & $  0.00486 \pm 0.00053$ & $  0.00526 \pm 0.00064$ & $  0.00490 \pm 0.00090 $\\
$    0.01$ & $  0.00970 \pm 0.00085$ & $  0.01029 \pm 0.00115$ & $  0.01056 \pm 0.00140 $\\
$    0.05$ & $  0.04874 \pm 0.00306$ & $  0.04894 \pm 0.00511$ & $  0.05228 \pm 0.00776 $\\
$     0.1$ & $  0.09785 \pm 0.00497$ & $  0.09956 \pm 0.00392$ & $  0.10351 \pm 0.01329 $\\
$     0.5$ & $  0.49601 \pm 0.01908$ & $  0.50897 \pm 0.02236$ & $  0.49371 \pm 0.04843 $\\
$       1$ & $  1.01719 \pm 0.03833$ & $  0.99922 \pm 0.02422$ & $  0.97016 \pm 0.07956 $\\
$       5$ & $  5.07901 \pm 0.23531$ & $  4.99453 \pm 0.17301$ & $  5.10122 \pm 0.38010 $\\
$      10$ & $ 10.08117 \pm 1.04427$ & $ 10.05533 \pm 0.35628$ & $ 10.32488 \pm 0.79659 $\\
\end{tabular}

}
\end{center}
\caption{The MLEs in the experiment with complete observed trajectories. The first column on the left shows the value of $\beta$ used for generating the paths and the two other columns show the mean $\pm$ the standard deviation of the MLEs over 10 repetitions.}
\label{tab:MultipleFullPaths}
\end{table}

\subsection{Estimation of edge costs}
\label{sec:EstimationOfEdgeCost}
So far we have only focused on the problem of estimating the inverse temperature parameter $\beta$ in situations where the graph structure, i.e.~the edge affinities (or the reference transition probabilities) and the edge costs are known a priori.
In this section, however, we consider briefly the problem of estimation of edge costs from trajectories.
Thus, assume that only the edge affinities (or the reference transition probabilities) are known a priori.
If the affinities cannot be quantified, but the graph structure (i.e.\ which nodes are connected by edges) is known, then we can simply assume affinities 
This assumption can be reasonable, as random walk behaviour can be easier to model than 
We then try to estimate the edge costs based on the observed trajectories.

It turns out that this estimation problem actually contains the problem of estimating $\beta$, meaning that solving the edge cost estimation problem also solves the estimation problem of $\beta$.
This is due to the fact that concerning the RSP distribution over paths (Equation (\ref{eq:GibbsBoltzmann})), parameter $\beta$ can be considered as a simple scaling factor of the path costs.
Namely, consider a data set of trajectories is generated according to the RSP distribution with a fixed inverse temperature, say $\beta^{*}$, which is unknown to the user.
But this is the same distribution as one would obtain by considering the RSP distribution with $\beta=1$, but on a modified graph with edge costs 
$\hat{c}_{ij} = \beta^{*} c_{ij}$, 
as then, for each $\wp$,
\begin{equation}
\label{eq:ScaledPathCostIsCostOfScaledEdgeCosts}
\beta^{*}\tc(\wp) 
= 
\beta^{*}
\!\! 
\sum_{l=1}^{L(\wp)} 
\!\! 
c_{\wp(l-1),\wp(l)} 
=
\sum_{l=1}^{L(\wp)} 
\!\! 
\hat{c}_{\wp(l-1),\wp(l)} 
=
\hat{c}(\wp).
\end{equation}
Expressed conversely, if we try to estimate the costs from the trajectories, we may simply assume that $\beta=1$, as a result of which we should end up with cost estimates $\hat{c}_{ij}$.
But these estimates contain the information of both the ``original'' edge costs and the value $\beta^{*}$ used to generate the data on the original graph.
This shows that the estimation of $\beta$ is in practice only a subproblem of the edge cost estimation problem. 

We now derive the MLE of edge costs in the case of complete trajectories, given a data set $\Omega$, as in Section~\ref{sec:CompleteTrajsMultipleST}.
For this, we assume the value $\beta=1$, and search for edge costs that maximize the likelihood of the trajectories.
\begin{theorem}
Given data set $\Omega = \cup_{s,t\in V} \Omega_{st}$, with $K_{st} = |\Omega_{st}|$, as in Section~\ref{sec:CompleteTrajsMultipleST}, and assuming that $\beta=1$, the MLE of the edge cost parameters $c_{ij}$, for all $(i,j) \in E$ is the value for which the expected number of visits over each edge $(i,j)$ over hitting $s$-$t$-paths, $\expVij(s,t)$, (from Equations (\ref{eq:ExpVij}) and (\ref{eq:ExpVComputation})) satisfy
\begin{equation}
\label{eq:MultipleSTFullPathsCostMLE}
\sum_{s,t \in V}
K_{st}
\expVij(s,t)
= 
\sum_{s,t \in V}
\sum_{\wp_{st} \in \Omega_{st}}
\Vij(\wp_{st}).
\end{equation}
\end{theorem}

\begin{proof}
Recall the log-likelihood of observing a set $\Omega_{st}$ of $K$ trajectories from $s$ to $t$, from Equation (\ref{eq:FullSinglePathLogLikelihood}) (now with $\beta=1$).
Using Equations~(\ref{eq:ExpVij}) and (\ref{eq:VisitsOverP}), the partial derivative of the log-likelihood w.r.t.\ an edge cost $c_{ij}$ is 
\begin{equation}
\label{eq:FullPathLogLCostDerivative}
\dfrac{\partial \logL}{\partial c_{ij}}
= 
K 
\expVij(s,t)
-
\sum_{\wp \in \Omega_{st}}
\Vij(\wp),
\end{equation}
where, for recollection, $\Vij(\wp)$ is the number of traversals over edge $(i,j)$ along path $\wp$ and $\expVij(s,t)$ is its expectation with respect to the RSP probability distribution over hitting paths from $s$ to $t$.

Setting the derivative in Equation~(\ref{eq:FullPathLogLCostDerivative}) to zero, we see that the MLE of an edge cost $c_{ij}$ should satisfy
\begin{equation}
\label{eq:FullPathsCostMLE}
\expVij(s,t)
= \dfrac{1}{K}
\sum_{\wp \in \Omega_{st}}
\Vij(\wp).
\end{equation}
The theorem follows directly by considering the data set as $\Omega = \cup_{s,t} \Omega_{st}$, with $K_{st} = |\Omega_{st}|$ where ${\Omega_{st} = \big\{ \wp \in \Omega \giv \wp \in \Pset \big\}}$ for all $s,t \in V$.
\end{proof}

Again, as Theorems~\ref{thm:CompleteTrajs_st} and \ref{thm:CompleteTrajs}, this result is analoguous to corresponding results in maximum entropy modelling~\cite{kapur1992entropy}, but appear here for the first time in the context of RSPs.
Unfortunately, as was the case with the MLE of $\beta$, the above cannot be solved analytically for $c_{ij}$.
Moreover, the dimension of the problem is now the number of edges, so the solution cannot be found only by a simple line search.
The MLE of the edge costs can, instead, be sought e.g.\ by performing a gradient ascent towards the direction given by the edge derivatives in Equation (\ref{eq:FullPathLogLCostDerivative}) above.
We leave the investigation of this idea for future work, and focus in the remainder of the paper on the estimation of $\beta$.

\section{Maximum likelihood estimation with incomplete trajectories}
\label{sec:IncompleteTrajectories}
This section deals with the case where data consists of \emph{incomplete trajectories} from $s$ to $t$.
For a path $\wp \in \Pset$, an incomplete trajectory means a subsequence of either the edge sequence or the node sequence constituting $\wp$.
For each such trajectory, we assume that the starting and target nodes, $s$ and $t$ are known and $t$ is considered as absorbing; this assumption is discussed in more depth later on.

In this section, we derive methods for computing the likelihood of the inverse temperature parameter $\beta$ given an observed data set $\Omega$ of incomplete trajectories, which constitute the most relevant contribution of this work.
These methods provide useful tools for many applications, because often in empirical tracking studies trajectories can only be recorded partially.
We derive the likelihood computation methods from the simplest case to the general one.
First we consider the cases where only one edge or one node is observed from one trajectory, and then generalize to cases where multiple edges or multiple nodes are observed over several trajectories.
We also demonstrate experimentally the behavior of the likelihood and validity of the methods.

The methods developed here for likelihood computation for incomplete trajectories are general in the sense that they make the most na\"{i}ve possible assumption on the method of observation, i.e.~the sampling of edges or nodes along the trajectories.
In more detail, we assume that the observed edges have been drawn from a uniform distribution over the actual trajectory.
Although simplistic, this assumption makes the method applicable to many different scenarios.
The sampling assumption is discussed in more depth in Section~\ref{sec:MultipleEdgeLikelihood}, where we also consider, as an alternative, that the number of observations along a trajectory is binomially distributed.
However, this binomial prior is not investigated in the experiments in the paper, but they rely on the uniformity assumption.

\subsection{Determining $s$ and $t$}
\label{sec:Determining_s_and_t}
In the case of incomplete trajectories, determining or inferring the starting node $s$ and the absorbing target node $t$ is not necessarily obvious.
In some cases it can be resolved trivially by determining the last observed node on the trajectory as the target.
In others, the target node might be evident in the network; for instance, considering movement in a computer game, the target node could be determined as the state where the game is compeleted or won.
Sometimes it can also be meaningful to select an area or a subset of nodes of the network as a set of absorbing targets.
Such a strategy can be more suitable in cases where the observed incomplete trajectories are sparse.

Another option, which we also use in the experiments in Section~\ref{sec:ApplicationToWildAnimalMovement}, is to consider, as in the previous strategy, an area of the network as the destination of movement, but decide the target $t$ for each trajectory as the first observed node of the trajectory within this area or subset (and neglect the remaining part of the trajectory).
The motivation for this is that the random walker is assumed to move according to the RSP distribution only until it reaches a particular part of the network, after which its movement behavior might change.
This assumption evidently holds for the use case in Section~\ref{sec:ApplicationToWildAnimalMovement}, where the data is based on seasonal migration of wild reindeer from one habitat to another.

For more sophisticated methods, one could also incorporate the \emph{bag-of-paths (BoP) models}~\citep{francoisse2017bag,lebichot2014semisupervised} in different ways.
The standard BoP framework is simply an extension of RSPs where the Gibbs-Boltzmann distribution is defined over the set of paths between all node pairs, instead of one $s$-$t$-pair at a time.
This formulation leads to the \emph{BoP probability} distribution over all $s$-$t$-pairs, meaning the probability that a given node-pair, $s$ and $t$, are the source and target nodes of a path sampled from the ``bag of paths''.
These probabilities are determined by the overall costs of paths between a given $s$ and $t$, compared to the costs of paths between other node-pairs.
A similar model has been considered in the context of transportation research in \citep{ryu2014dual}.
An even more sophisticated approach would be the \emph{margin-constrained BoP} model~\citep{guex2019randomized}, where two distributions on the nodes of the network, $k \in V$, are given as input; one fixes the probability that $k$ is a source node of a path, while the other fixes the probability that $k$ is a target node of a path.
The model then computes flows based on these fixed margins and the Gibbs-Boltzmann distribution over the set of all (hitting or regular) paths from the sources to the targets.
In other words, the margin-constrained BoP model is applicable when the source and target distributions can be determined from data.
Finally, we leave the consideration of incorporating these ideas in the maximum likelihood model for future research.

\subsection{One observed edge}
\label{sec:OneEdgeLikelihood}
Consider a situation where the data set $\Omega$ consists of only one edge $(i,j)$ that has been observed from one trajectory from $s$ to $t$, where $s$ and $t$ are known beforehand.
In this section, we present a computable expression of the likelihood function of $\beta$ in such a case, discuss the intuition behind this expression, and also present an example illustrating the behavior of the likelihood function.

\subsubsection{Derivation of the likelihood function}
\begin{theorem}
\label{thm:OneEdgeLikelihood}
Let $\wp$ be a hitting path from $s$ to $t$ drawn from the RSP distribution (Equation~(\ref{eq:GibbsBoltzmann})), and let $i$ be a node sampled from a uniform distribution over the node sequence of $\wp$.
The likelihood function of $\beta$ in such a situation can be computed as
\begin{equation}
\label{eq:IncompletePathsLikelihood}
\Lhood(\beta \giv (i,j); s, t) 
= \dfrac{\mleft[ 
  \mathbf{L}_{ij} 
  \mright]_{s,n+t}}{\PF},
\end{equation}
where $\mleft[\mathbf{L}_{ij}\mright]_{s,n+t}$ is the element $(s, n+t)$ of the $(2n \times 2n)$ matrix that can be expressed using the \emph{matrix logarithm} \citep{higham2008functions}, $\mathbf{log}$, as
\begin{equation}
\mathbf{L}_{ij} 
= 
\sum_{k=1}^{\infty} \frac{\mathbf{Q}_{ij}^{k}}{k} 
= 
-\mathbf{log}(\mathbf{I} - \mathbf{Q}_{ij}),
\end{equation}
where $\mathbf{Q}_{ij}$ is the $(2n \times 2n)$ block matrix 
\begin{equation}
\mathbf{Q}_{ij} 
= 
\mleft[
\begin{array}{cc}
\Wt
& 
\mathbf{W}_{ij} 
\\
\mathbf{0}
&
\Wt
\end{array}
\mright].
\label{eq:QijDef}
\end{equation}
Recall that $\Wt$ is the matrix from Equation~(\ref{eq:Wt}) obtained by setting row $t$ of matrix $\mathbf{W}$ to zero, and $\mathbf{W}_{ij}$ is the matrix containing value $w_{ij}$ at element $(i,j)$ and zero elsewhere.
\end{theorem}

\begin{proof}
Let us denote by $\rho$ the random variable corresponding to the drawing of a trajectory from $\Pset$ along which the edge is observed.
We then make the fundamental assumption that the probability distribution of $\rho$ is the RSP distribution (see Equations~(\ref{eq:GibbsBoltzmann}) and (\ref{eq:PathLikelihood})): \begin{equation}
\mathrm{P}(\rho = \wp)
=
\ProbRSP(\wp)
=
\dfrac{\tw(\wp)}{\PF}
.
\end{equation}
Furthermore, let $\varepsilon$ be the random variable corresponding to the drawing of the observed edge.
As discussed above, assume that $\varepsilon$ is sampled uniformly from the edge sequence of the observed trajectory.
Then the conditional probability of drawing edge $(i,j)$, given a particular path $\wp$, is given by the number of times $\wp$ traverses edge $(i,j)$, $\Vij$, divided by the total number of edges traversed along the $\wp$, i.e.~its length, $L(\wp)$:
\begin{equation}
\mathrm{P}
(\varepsilon = (i,j) | \rho = \wp)
= 
\dfrac{\Vij(\wp)}{L(\wp)}.
\label{eq:ConditionalEdgeProbability}
\end{equation}
Thus, the joint probability of observing edge $(i,j)$ along path $\wp$ is
\begin{equation}
\label{eq:SingleEdgeJointProbability}
\mathrm{P}
\mleft(
\varepsilon = (i,j), \rho = \wp
\mright) 
= 
\mathrm{P}
\mleft(
\rho = \wp
\mright) 
\mathrm{P}
\mleft(
\varepsilon = (i,j
\mright) 
\giv 
\rho = \wp)
= 
\dfrac{\tw(\wp) \Vij(\wp)}{\PF L(\wp)}.
\end{equation}

Note that the partial derivative of the path likelihood $\widetilde{w}(\wp)$ (Equation~(\ref{eq:PathLikelihood})) can be expressed, using Equation~(\ref{eq:VisitsOverP}), as
\begin{equation}
\label{eq:wDerivative}
\frac{\partial \widetilde{w}(\wp)}{\partial c_{ij}}  
= 
\Prw(\wp) 
\frac{\partial \exp(-\beta \tc(\wp))}{\partial c_{ij}} 
= 
-\beta 
\widetilde{w}(\wp) 
\Vij(\wp).
\end{equation}
Based on this, the probability of observing edge $(i,j)$, which is also the likelihood function of $\beta$, is obtained from (\ref{eq:SingleEdgeJointProbability}) by marginalizing out $\rho$:
\begin{align}
\Lhood(\beta \giv (i,j); s, t) 
&= \mathrm{P}(\varepsilon = (i,j) ; \beta)
   \\
&= \Psum
   \mathrm{P}(\varepsilon = (i,j), \rho = \wp; \beta)
  \\
&= \dfrac{1}{\PF} 
   \sum_{\wp \in \Pset} 
   \dfrac{\widetilde{w}(\wp) \Vij(\wp)}
   {L(\wp)} 
= 
-\dfrac{1}{\beta \PF} 
\frac{\partial \widetilde{\mathcal{Z}}_{st}}{\partial c_{ij}},
\label{eq:OneEdgeLikelihoodLast}
\end{align}
where we define a new partition function
\begin{equation}
\label{eq:NewPartitionFunction}
\widetilde{\mathcal{Z}}_{st} 
= \sum_{\wp \in \Pset} 
  \dfrac{\widetilde{w}(\wp)}{L(\wp)}
= \sum_{k=1}^{\infty}
  \sum_{\wp_{k} \in \Pset^{(k)}} 
  \dfrac{\widetilde{w}(\wp_{k})}{k},
\end{equation}
with $\Pset^{(k)}$ being the set of hitting paths from $s$ to $t$ of length $k$.
Recalling Equation (\ref{eq:WPowers}) and that $\Wt$ is the matrix obtained from matrix $\mathbf{W}$ by setting row $t$ to zero, we see that the new partition function can be expressed using the matrix logarithm:
\begin{equation}
\label{eq:NewPFComputation}
\widetilde{\mathcal{Z}}_{st}
= 
\mleft[
\Wt 
+ 
\frac{\Wt^2}{2} 
+ 
\frac{\Wt^3}{3} 
+ 
\cdots 
\mright]
_{st}
= 
\mleft[ 
-\mathbf{log}
(\mathbf{I}-\Wt) 
\mright]
_{st},
\end{equation}
which exists, as $\rho(\Wt) < 1$.

Next we show how the derivative of the above quantity, needed for the computation of the log-likelihood (\ref{eq:OneEdgeLikelihoodLast}), can be computed. 
For this, we use the shorthand notation $\partial_{c_{ij}} X = \partial X / \partial c_{ij}$ for the partial derivative of a quantity $X$ (e.g.\ a function or a matrix).
First of all, for any $(i,j) \in E$ such that $i\neq t$,
\begin{equation}
\label{eq:WDerivative}
\partial_{c_{ij}} 
\Wt 
= 
\mleft(
\partial_{c_{ij}} 
w_{ij} 
\mright) 
\mathbf{I}_{ij} 
= 
\mleft(
\partial_{c_{ij}} 
\prw_{ij} 
\exp(-\beta c_{ij}) 
\mright) 
\mathbf{I}_{ij} 
=
-\beta
\mathbf{W}_{ij},
\end{equation}
where $\mathbf{I}_{ij} = \mathbf{e}_{i}\mathbf{e}_{j}^{\mathsf{T}}$ is the matrix with 1 at element $(i,j)$ and zero elsewhere, and where we have defined $\mathbf{W}_{ij} = w_{ij}\mathbf{I}_{ij}$.
We can thus write the derivative of $\Wt^{k}$ as
\begin{align}
\frac{\partial \Wt^k}{\partial c_{ij}} 
&= 
(\partial_{c_{ij}} \Wt) 
\Wt^{k-1}
+ 
\Wt 
(\partial_{c_{ij}} 
\Wt^{k-1}) 
\nonumber
\\
&= 
-\beta 
\mleft( \mathbf{W}_{ij} \Wt^{k-1} + \Wt 
\mleft( \mathbf{W}_{ij} \Wt^{k-2} + \Wt 
\mleft( \mathbf{W}_{ij} \Wt^{k-3} + \cdots \mright) 
\mright) 
\mright) 
\nonumber
\\
&= 
-\beta 
\sum_{l=1}^{k} 
\Wt^{l-1} 
\mathbf{W}_{ij} 
\Wt^{k-l}
\triangleq 
-\beta 
\mathbf{S}_{ij}^{(k)},
\label{eq:WkDerivative}
\end{align}
where we have defined, for convenience,
\begin{equation}
\mathbf{S}_{ij}^{(k)} 
= 
\sum_{l=1}^{k} 
\Wt^{l-1} 
\mathbf{W}_{ij} 
\Wt^{k-l}.
\label{eq:SijkDef}
\end{equation}
This matrix can be computed with the help of the auxiliary $(2n \times 2n)$ block matrix
\begin{equation}
\mathbf{Q}_{ij} 
= 
\mleft[
\begin{array}{cc}
\Wt
& 
\mathbf{W}_{ij} 
\\
\mathbf{0}
&
\Wt
\end{array}
\mright],
\end{equation}
as the $k$-th power of this matrix can be shown to be:
\begin{equation}
\mathbf{Q}_{ij}^{k} 
= 
\mleft[
\begin{array}{cc}
\Wt^{k} 
& 
\mathbf{S}_{ij}^{(k)} 
\\
\mathbf{0} 
& 
\Wt^{k} 
\end{array}
\mright],
\label{eq:QijPowerDef}
\end{equation}

Accordingly, the derivative of the new partition function, required in Equation~(\ref{eq:OneEdgeLikelihoodLast}), can be computed as
\begin{align}
\dfrac{\partial \widetilde{\mathcal{Z}}_{st}}{\partial c_{ij}  }
&= 
\mathbf{e}_{s}^{\mathsf{T}} 
\sum_{k=1}^{\infty} 
\mleft(
\dfrac{
\partial 
\Wt^{k} / k 
}{\partial c_{ij}  }
\mright)
\mathbf{e}_{t} 
\label{eq:PFDerivativeFirst}
\\
&= - \beta \mathbf{e}_{s}^{\mathsf{T}} \sum_{k=1}^{\infty} \frac{\mathbf{S}_{ij}^{(k)}}{k} \mathbf{e}_{t} \\
&= -\beta 
   \mathbf{e}_{s}^{\mathsf{T}} 
   \sum_{k=1}^{\infty} 
   \frac{\mathbf{Q}_{ij}^{k}}{k} 
   \mathbf{e}_{n+t}
   \label{eq:PFDerivativeIntermediate}
   \\
&= -\beta 
   \mathbf{e}_{s}^{\mathsf{T}} 
   \mathbf{L}_{ij}
   \mathbf{e}_{n+t} 
   \\
&= -\beta 
   \mleft[ \mathbf{L}_{ij} \mright]_{s, n+t},
\label{eq:PFDerivative}
\end{align}
where 
\begin{equation}
\mathbf{L}_{ij} = \sum_{k=1}^{\infty} \frac{\mathbf{Q}_{ij}^{k}}{k} = -\mathbf{log}(\mathbf{I} - \mathbf{Q}_{ij}).
\end{equation}

Finally, combining (\ref{eq:OneEdgeLikelihoodLast}) and (\ref{eq:PFDerivative}), the likelihood function can be expressed as
\begin{equation}
\Lhood(\beta \giv (i,j); s, t) 
= -\dfrac{1}{\beta \PF} 
  (-\beta 
  \mleft[ \mathbf{L}_{ij} \mright]_{s,n+t})
= \dfrac{\mleft[ \mathbf{L}_{ij} \mright]_{s,n+t}}{\PF}.
\end{equation}
\end{proof}

It is possible to derive the derivative with respect to $\beta$ of the above likelihood function.
However, the form of the derivative is rather complicated, involving the matrix logarithm of a $(4n \times 4n)$ block matrix.
More importantly, the root of the derivative cannot be solved in closed form, unlike in the case with complete trajectories in Section~\ref{sec:CompleteTrajectories}.
Thus, for finding $\betaMLE$, it is more straightforward to search directly for a value of $\beta$ that maximizes (\ref{eq:IncompletePathsLikelihood}).

\subsubsection{Intuition behind the computation of the likelihood function}

The computation of the single-edge likelihood can also be understood from another, more informal, perspective.
Namely, the matrix $\mathbf{Q}_{ij}$ in Equation~(\ref{eq:QijDef}) can be interpreted as a likelihood matrix defining a new graph consisting of the original graph $G$ (with links going out of the target node $t$ deleted) augmented with its copy $G'$ and with a directed edge from node $i$ of subgraph $G$ to node $j$ of subgraph $G'$ (i.e.~node $n+j$ of the new graph) with edge likelihood $w_{ij}$.
This idea is illustrated in Figure~\ref{fig:OneCopyGraph}.
Element $(s,n+t)$ of the $k$-th power of matrix $\mathbf{Q}_{ij}$ thus enumerates all $s$-$t$-paths of length $k$ that traverse edge $(i,j)$ and cumulates the corresponding path likelihoods.

Concerning the computation in practice, the current standard for matrix logarithm computations is the algorithm developed in~\citep{al2013computing}, which is implemented, for instance, in Matlab and SciPy as the \texttt{logm} function.
However, we are only interested in computing one element of the matrix logarithm, and for this purpose we found the Matlab implementation numerically too unreliable.
More exactly, the algorithm produces inaccurate values of the element $(s, n+t)$ of matrix $\mathbf{L}_{ij}$ (see Equation (\ref{eq:PFDerivative})) when $\beta$ is relatively large, causing the values $[\mathbf{L}_{ij}]_{s,n+t}$ to be very small.
To get more reliable values, we compute the matrix logarithm element iteratively based on the power series expression of the logarithm, by computing a finite sum according to Equation (\ref{eq:NewPFComputation}). 

For future research, it is worth mentioning that the above computational techniques involving the matrix logarithm could also be used, for example, for defining a network centrality measure based on the probability of uniformly sampling an edge or node over paths between nodes.
This could be done considering the RSP distribution, but also with the natural random walk distribution.
Namely, for the observation probability over natural random $s$-$t$-walks, one can simply perform the above computations by setting $\beta=0$, which corresponds to replacing the matrix $\Wt$ with the matrix $\PRW - \mathbf{I}_{tt}\PRW$, i.e.~the random walk transition probability matrix with row $t$ set to zero.

\subsubsection{Example of the one-edge likelihood}

We demonstrate the behavior of the one-edge likelihood, expressed in Equation~(\ref{eq:IncompletePathsLikelihood}), with a simple example illustrated in Figure \ref{fig:ExampleGrid}.
We consider a $5 \times 5$ grid with links to diagonal neighbors included.
The edge costs are $1$ for horizontal and vertical edges, and $\sqrt{2}$ for diagonal edges and affinities as inverse costs, $a_{ij} = 1/c_{ij}$.
We consider paths from node 7, in the lower left corner, to node 25 in the upper and right-most corner.
We then compute, for different values of $\beta$, the likelihood of observing each of the edges leaving from node 7 as an edge of such a path.

\begin{figure}
\centering
\includegraphics[width=0.6\textwidth]{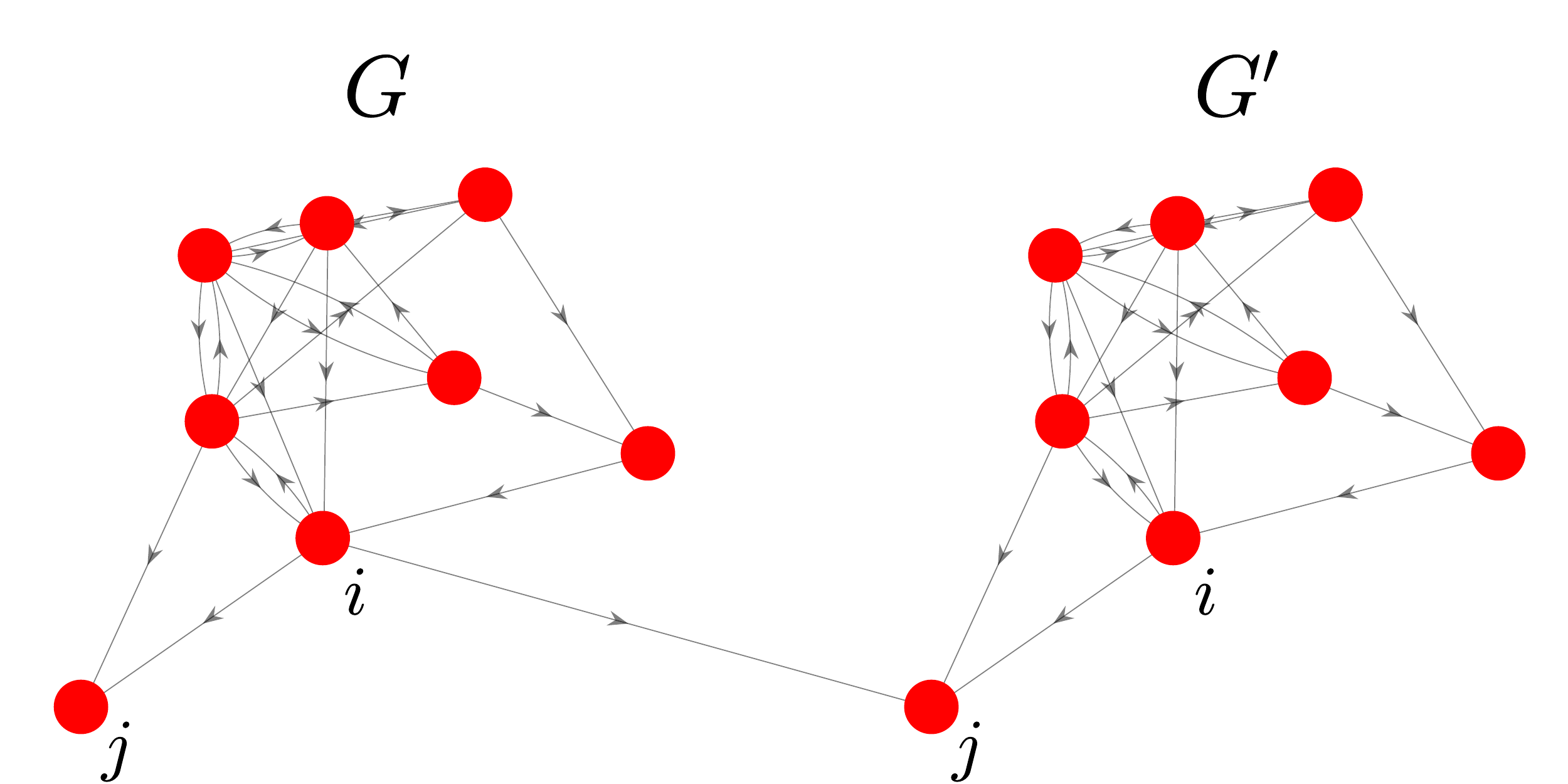}
\caption{The computation of the one-edge likelihood can be interpreted as making a copy $G'$ of the original graph $G$ and adding an edge from node $i$ of $G$ to node $j$ of $G'$ and considering likelihoods of paths from node $s$ of $G$ to node $t$ of $G'$.}
\label{fig:OneCopyGraph}
\end{figure}

The results are plotted in Figure \ref{fig:GridExamplePlots}, which shows the likelihoods for each edge as a function of $\beta$ as separate curves.
Moreover, the dashed vertical lines indicate the peaks of the curves, i.e.~the value $\betaMLE$ corresponding to observing the edge in question.
As is expected, the highest $\betaMLE$ value is obtained for the edge $(7,13)$, which lies on the least cost path from node $7$ to node $25$.%
\footnote{In fact, in this case, the likelihood can be seen to increase indefinitely, indicating that $\betaMLE = \infty$. Indeed, as $\beta$ increases, the RSP probability of the least cost path along edge $(7,13)$ increases, and the probability of observing the edge $(7,13)$ increases likewise. This also leads to the increase of the likelihood function indefinitely as $\beta$ increases.}
In addition, the $\betaMLE$ values decrease as we consider edges that move more and more away from the shortest path.

\begin{figure}
	\centering
	\begin{subfigure}[b]{0.42\textwidth}
		\includegraphics[width=\textwidth]{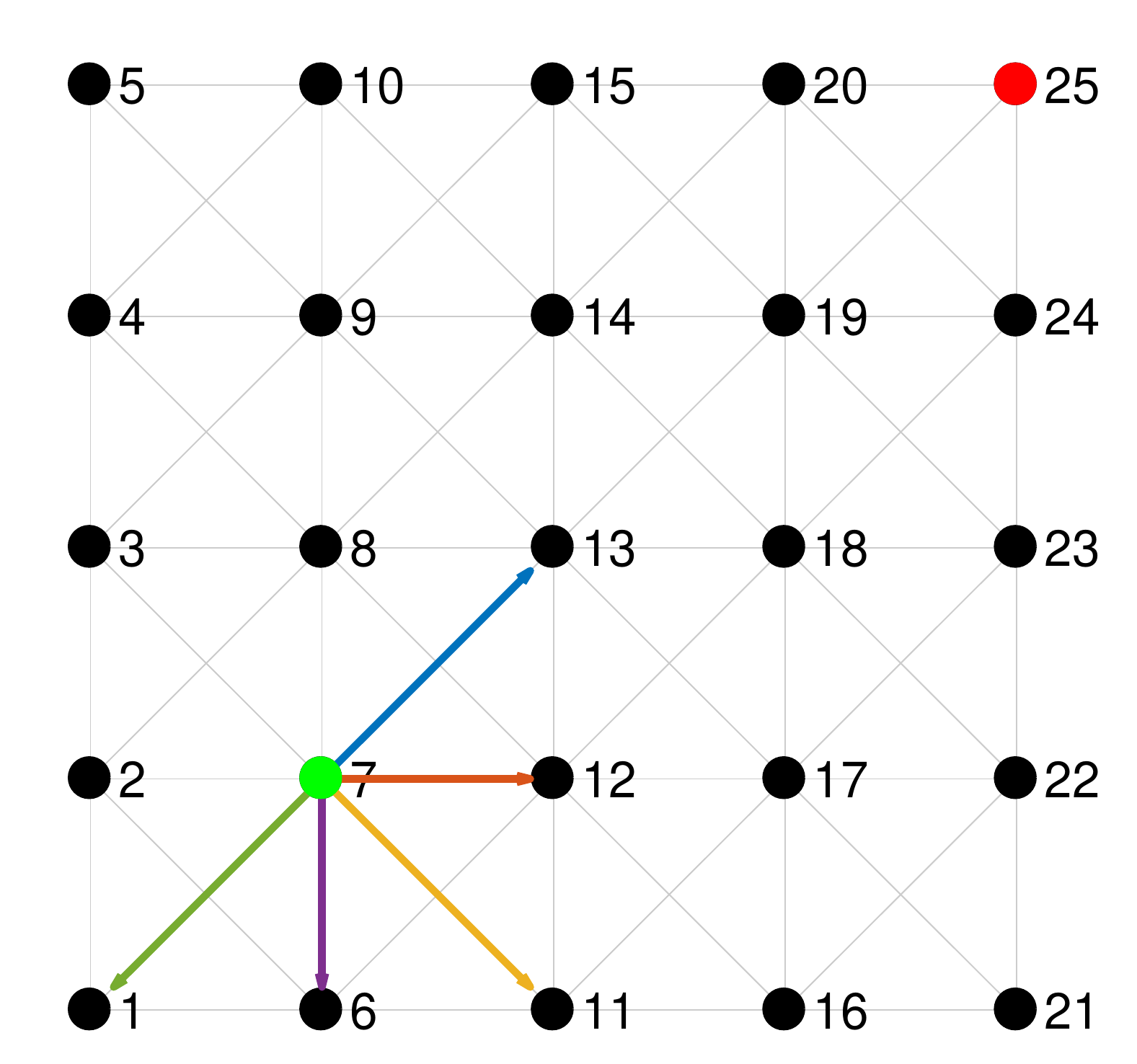}
		\caption{The grid}
		\label{fig:ExampleGrid}
	\end{subfigure}
	\begin{subfigure}[b]{0.57\textwidth}
		\includegraphics[width=\textwidth]{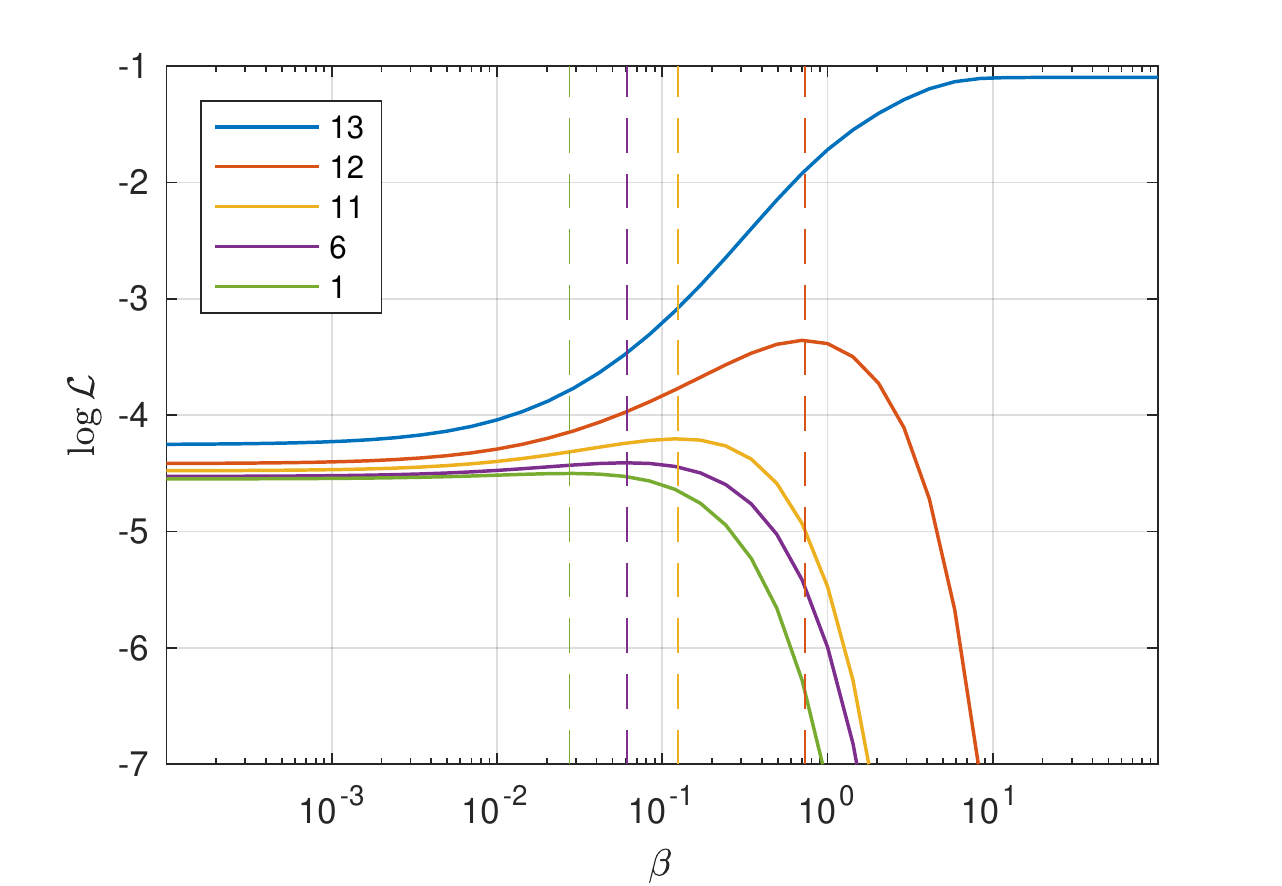}
		\caption{The likelihood functions w.r.t.\ $\beta$}
		\label{fig:GridExamplePlots}
	\end{subfigure}
	\caption{The $5 \times 5$ grid, where the green node 7 is the starting node, $s$, the red node 25 is the absorbing target node, $t$ (A); the log-likelihoods of observing each of the coloured edges for different values of $\beta$ (B). Note the logarithmic scale on the horizontal axis.}
	\label{fig:GridExample}
\end{figure}

\subsection{One observed node}
\label{sec:OneNodeLikelihood}

The likelihood of $\beta$ can also be computed in the case where a trajectory is observed to visit a particular node instead of an edge.
We derive the \emph{one-node likelihood} using the fact that the number of visits to a node $i$ along a path $\wp$, $n_{i}(\wp)$, can be expressed as the sum of travelsals out of that node: $n_{i}(\wp) = \sum_{j \in Succ(i)} \Vij(\wp)$.
Note that this definition of $n_{i}(\wp)$ directly means that the observed node cannot be the absorbing terminal node $t$, i.e.~that the observation is made before the end of the path.

However, the possibility that the intermediate node, drawn from a path $\wp \in \Pset$, is actually node $s$, must be considered more carefully.
In this work, we assume that the outcome of the observation can indeed be node $s$, but that the observation is done \emph{after the first step}, i.e.~that the observation is made from the subsequence $\wp(1, \ldots, L(\wp) - 1)$, excluding the first node $\wp(0)=s$ and the last node $\wp(L(\wp)) = t$.
Moreover, we assume that the node is sampled uniformly from this subsequence.
The exclusion of the first node $\wp(0)$ from consideration complicates the derivation a bit compared to the derivation of the single-edge likelihood.
Because of this, the full derivation of the result is presented separately in Appendix~\ref{app:OneNodeLikelihood}.

\begin{theorem}
\label{thm:OneNodeLikelihood}
Let $\wp$ be a hitting path from $s$ to $t$ drawn from the RSP distribution (Equation~(\ref{eq:GibbsBoltzmann})), and let $i$ be a node sampled from a uniform distribution over the node sequence of $\wp$.
Then, the likelihood function of $\beta$ can be computed as
\begin{equation}
\Lhood(\beta \giv i; s, t)
=
\dfrac{1}{\PF}
\sum_{u \in Succ(s)}
w_{su}
\mleft[
\mathbf{L}_{i}
\mright]_{u, n+t},
\label{eq:OneNodeLikelihoodFinal}
\end{equation}
which contains, again, a matrix logarithm,
$\mathbf{L}_{i} = -\mathbf{log}(\mathbf{I}-\mathbf{Q}_{i})$,
where
\begin{equation}
\label{eq:QiDef}
\mathbf{Q}_{i} = 
\mleft[\begin{array}{cc}\Wt & \mathbf{W}_{\mathrm{r}(i)} \\
\mathbf{0} & \Wt 
\end{array}\mright],
\end{equation}
and 
$\mathbf{W}_{\mathrm{r}(i)} = \mathbf{I}_{ii} \mathbf{W}$
is the matrix containing the $i$-th row of matrix $\mathbf{W}$ on its $i$-th row, but zeros elsewhere.
\end{theorem}
\begin{proof}
In Appendix~\ref{app:OneNodeLikelihood}
\end{proof}

\subsection{Multiple observed edges}
\label{sec:MultipleEdgeLikelihood}

In many data sets of incomplete trajectories, the trajectories contain several observations, meaning that the trajectories have been observed to pass through more than only one edge or node of the network.
Here, the computation of the likelihood of $\beta$ given such a sequence of multiple observed edges is presented.
For this, consider a data set containing one sequence
$
\tilde{e} = (e_{1}, \ldots, e_{M})
$
of $M$ edges observed from one trajectory (in the corresponding order).
Let us also denote $(i_m, j_m) = e_m$ for all $m = 1,\ldots, M$ and assume that all $i_m$ and $j_m$ (except possibly $j_M$) are different from the target node $t$. 

One main difference with the results presented earlier, for the one-edge and one-node likelihoods, is that now the number of observations from a trajectory also needs to be considered as a random variable.
Let us denote this random variable by $\mu$.
Furthermore, let $\tilde{\varepsilon} = (\varepsilon_{1}, \ldots, \varepsilon_{M})$ be the random vector corresponding to the observed edge sequence of given length $M$ consisting of the random variables corresponding to the observed individual edges. 
Similarly to the uniformity assumption used in the cases of observing only one edge or node, we also assume a uniform distribution both for $\mu$ and for $\tilde{\varepsilon}$, in addition to which we assume independence between the random variables $\varepsilon_{1}, \ldots, \varepsilon_{M}$.
One justification for considering the uniform distribution in both cases is that they are, in this setting, the maximum entropy distributions, which implies that these assumptions are the most na\"ive and generic ones possible.

In Section~\ref{sec:Binomial} we will also briefly consider an alternative where $\mu$ is binomially distributed.
Assuming other distributions may also make the model more accurate in cases where the actual sampling is known to behave in a certain way, but the na\"ive uniformity assumption can be useful when the sampling process is not known or is very irregular.
Such is the case, for instance, with the real data example in Section \ref{sec:ApplicationToWildAnimalMovement}, where we estimate $\beta$ when fitting the RSP model to a data set of incomplete trajectories of wild reindeer.
In that data, although the locations of the animals are mostly measured at constant time intervals, there are also missing observations causing long gaps between measurements.

Let us now derive the likelihood using the uniformity assumptions. 
Thus, assume that the number of oberved edges, given a path $\wp$, is uniformly distributed, i.e.~%
\begin{equation}
\label{eq:MuConditional}
\mathrm{P}(\mu = M \giv \rho = \wp) = 1/L(\wp).
\end{equation}
Similarly, assume that given a path $\wp$ and a number of observations $M$ the probability of observing the edge sequence $\tilde{e}$ of $M$ edges from $\wp$ is uniformly distributed over all $\binom{L(\wp)}{M}$ possible subsequences of $M$ edges that can be drawn from $\wp$.
For this, we denote by 
$
\V_{\tilde{e}}(\wp)
$
the number of times that the edge sequence of $\wp$ contains $\tilde{e}$ as a subsequence.
Then, the distribution of $\tilde{\varepsilon}$, conditional on the number of observations, is
\begin{equation}
\label{eq:VarepsilonConditional}
\mathrm{P}
(
\tilde{\varepsilon} 
= 
\tilde{e} 
\giv 
\mu = M, 
\rho = \wp
) 
= 
\dfrac{\V_{\tilde{e}}(\wp)}{\binom{L(\wp)}{M}}.
\end{equation}

\begin{theorem}
\label{thm:MultiEdgeLikelihood}
Let $\wp$ be a hitting path from $s$ to $t$ drawn from the RSP distribution (Equation~(\ref{eq:GibbsBoltzmann})), and let $\tilde{e}$ be a sequence of $M$ edges, sampled from a uniform distribution over the edge sequence of $\wp$, where, furthermore, $M$  is drawn from a uniform distribution over $\{1, \ldots, L(\wp)\}$.
Then, the likelihood function of $\beta$ can be computed as
\begin{equation}
\label{eq:MultipleEdgeFullLhood}
\Lhood(\beta \giv \tilde{e}; s, t) 
= 
\dfrac{
\mleft[ 
\mathbf{L}_{\tilde{e}} 
\mright]_{s,Mn + t}}{\PF}.
\end{equation}
where
\begin{equation}
\mathbf{L}_{\tilde{e}} 
= 
\sum_{k=M}^{\infty} 
\dfrac{\mathbf{Q}_{\tilde{e}}^k}{\binom{k}{M}k},
\label{eq:OmegaMatrixLogarithm}
\end{equation}
and
\begin{equation}
\label{eq:QOmegaDef}
\mathbf{Q}_{\tilde{e}}
= \mleft[\begin{array}{cccccc}
\Wt & \mathbf{W}_{i_{1}j_{1}} & \mathbf{0} & \mathbf{0} & \ldots & \mathbf{0}\\
\mathbf{0} & \Wt & \mathbf{W}_{i_{2}j_{2}} & \mathbf{0} & \ldots & \mathbf{0}\\
\mathbf{0} & \mathbf{0} & \Wt & \ddots &  & \vdots \\
\vdots & \vdots & \ddots &  \ddots & & \mathbf{0} \\
 & & &  \mathbf{0} & \Wt &  \mathbf{W}_{i_{M}j_{M}}\\
\mathbf{0} & \ldots &  &  \ldots & \mathbf{0} & \Wt
\end{array}\mright].
\end{equation}

\end{theorem}

\begin{proof}
We present an exact proof of the theorem only in the special case of observing two edges in Appendix~\ref{app:TwoEdgeLikelihood}. 
Below, we sketch the intuition behind the proof in an informal way.
\end{proof}

Using Equations~(\ref{eq:MuConditional}) and (\ref{eq:VarepsilonConditional}) the likelihood  can be written as
\begin{align}
\Lhood(\beta \giv \tilde{e}; s, t)
&=
\mathrm{P}
\mleft(
\tilde{\varepsilon}
= 
\tilde{e},
\mu = M
\mright)
\\
&=
\psum
\Prob(\wp)
\mathrm{P}
\mleft(
\tilde{\varepsilon} = 
\tilde{e}
\giv
\mu = M,
\rho = \wp
\mright)
\mathrm{P}
(\mu = M
\giv
\rho=\wp)
\label{eq:MultiEdgeLikelihood2}
\\
&=
\psum
\Prob(\wp)
\dfrac{\V_{\tilde{e}}(\wp)}{\binom{L(\wp)}{M}}
\cdot
\dfrac{1}{L(\wp)}
\\
&=
\dfrac{1}{\PF}
\sum_{k=M}^{\infty}
\dfrac{
1
}{
\binom{k}{M}k
}
\sum_{\wp_{k} \in \Pset^{(k)}}
\widetilde{w}(\wp_{k})
\V_{\tilde{e}}(\wp_{k})
.
\label{eq:MultipleEdgeLikelihood}
\end{align}
The exact manipulation of the above expression is fairly involved, and is presented in Appendix~\ref{app:TwoEdgeLikelihood} only for the case of observing two edges on a trajectory.
The derivation for two edges extends naturally to the case of an arbitrary number of observed edges, although the proof of this is omitted for brevity.
Here we only give an intuitive description of the derivation of the multiple-edge likelihood based on the idea of copying the graph, used also earlier for explaining the computation of the one-edge likelihood, in Section~\ref{sec:OneEdgeLikelihood} and illustrated in Figure~\ref{fig:OneCopyGraph}.

\begin{figure}
\centering
\includegraphics[width=0.95\textwidth]{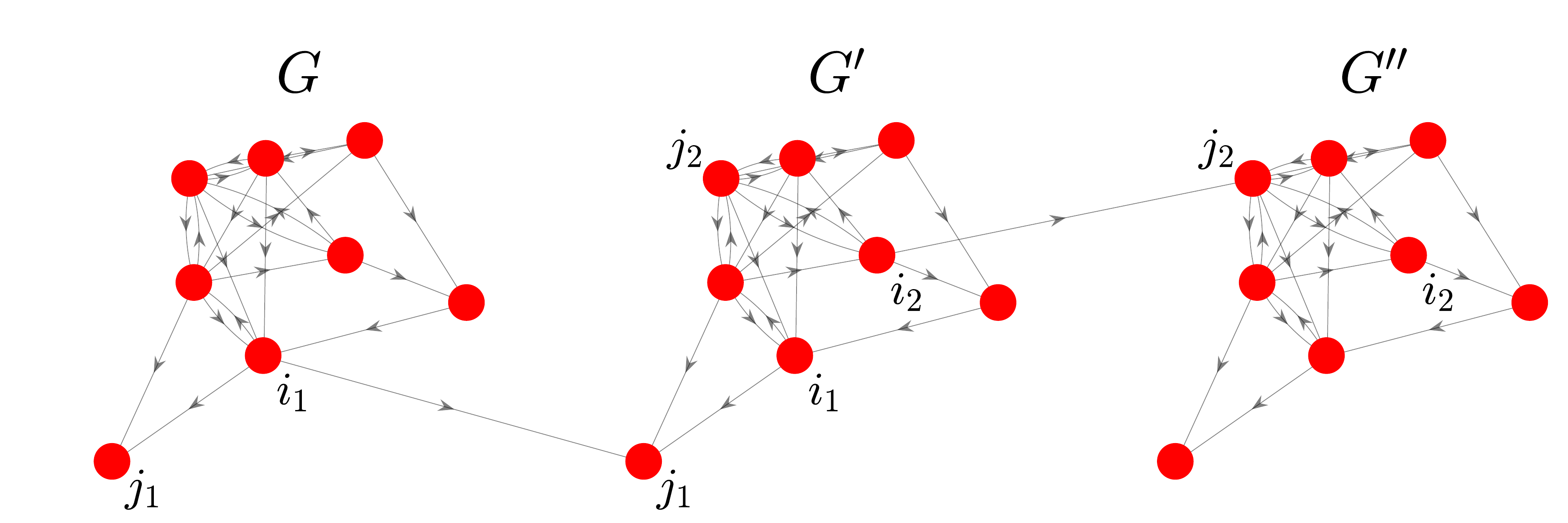}
\caption{The computation of the two-edge likelihood can be interpreted as making two copies, $G'$ and $G''$, of the original graph $G$ and adding the required edges between the copies.}
\label{fig:TwoCopiesGraph}
\end{figure}

Namely, for the edge sequence $\tilde{e}$, we may consider producing $M$ duplicates of the original graph $G$ and connect the duplicates via the edges of the sequence $\tilde{e}$. 
This is illustrated for the case of two observed edges in Figure~\ref{fig:TwoCopiesGraph}.
Indeed, the extended likelihood matrix $\mathbf{Q}_{\tilde{e}}$ from Equation~(\ref{eq:QOmegaDef}) represents such an extended graph of $M+1$ copies of the original graph.

Now, it can be shown that the inner sum in the expression in (\ref{eq:MultipleEdgeLikelihood}) is given by element $(s,Mn + t)$ of the $k$-th power of the matrix 
$\mathbf{Q}_{\tilde{e}}$; more exactly
\begin{equation}
\label{eq:Q_e_power}
\sum_{\wp_{k} \in \Pset^{(k)}}
\tw(\wp_{k}) \Vij(\wp_{k})
=
\mleft[
\mathbf{Q}_{\tilde{e}}^{k}
\mright]_{s, Mn + t}.
\end{equation}
This result is proved formally for the two-edge likelihood in Appendix~\ref{app:TwoEdgeLikelihood}.
The case of observing $M$ edges follows from there naturally.
Inserting the above Equation~(\ref{eq:Q_e_power}), as well the definition in Equation~(\ref{eq:OmegaMatrixLogarithm}), into Equation~(\ref{eq:MultipleEdgeLikelihood}) leads to the desired result.

Note that unlike in the one-edge or one-node likelihood cases, the expression of the likelihood in Theorem~\ref{thm:MultiEdgeLikelihood} has no direct connection with the matrix logarithm, and that the expression in Equation~(\ref{eq:OmegaMatrixLogarithm}) cannot be written in closed form.
Instead, we compute the likelihood by computing the sum in Equation~(\ref{eq:OmegaMatrixLogarithm}) up to convergence.

\subsection{Multiple observed nodes}
\label{sec:MultipleNodeLikelihood}
The likelihood of $\beta$ given an incomplete node-trajectory can be computed similarly to the likelihood of the incomplete edge-trajectory.
The derivation of the computation is similar to the extension from the one-edge likelihood to the one-node likelihood presented in Appendix~\ref{sec:OneNodeLikelihood}. 
Because of this, we omit the derivation and only state the result.
\begin{theorem}
Again, let $\wp$ be a hitting path from $s$ to $t$ drawn from the RSP distribution (Equation~(\ref{eq:GibbsBoltzmann})), and now, let $\tilde{v}$ be a sequence of $M$ nodes, sampled from a uniform distribution over the node sequence of $\wp$ (excluding the first and last nodes of the sequence)\footnote{Again, as in the one-node likelihood case, studied in Section~\ref{sec:OneNodeLikelihood}, the possibility of the first node being $s$ must be considered separately.
As was done there, we here also make the assumption that the first observed node can be $s$, but it cannot be the first node, $\wp(0)$, of the observed path $\wp$. 
}, where, furthermore, $M$  is drawn from a uniform distribution over $\{1, \ldots, L(\wp)-1\}$.
With these assumptions, the multiple-node likelihood is given, similarly to Equations~(\ref{eq:OneNodeLikelihoodFinal})~and~(\ref{eq:MultipleEdgeFullLhood}), by
\begin{equation}
\label{eq:MultipleNodeLikelihood}
\Lhood(\beta \giv \tilde{v}) 
= 
\dfrac{1}{\PF}
\sum_{h \in Succ(s)}
w_{su}
\mleft[ \mathbf{L}_{\tilde{v}} \mright]_{h,Mn + t},
\end{equation}
where the matrix $\mathbf{L}_{\tilde{v}}$ is defined equivalently to Equation (\ref{eq:OmegaMatrixLogarithm}), but by replacing matrix $\mathbf{Q}_{\tilde{e}}$ with matrix $\mathbf{Q}_{\tilde{v}}$, which, furthermore, is defined equivalently to the definition of the matrix $\mathbf{Q}_{\tilde{e}}$ in Equation~(\ref{eq:QOmegaDef}), however, replacing the matrices 
$\mathbf{W}_{i_{m}j_{m}}$ with the node-based matrices 
$\mathbf{W}_{\mathrm{r}(i_{m})} = \mathbf{I}_{i_{m}i_{m}} \mathbf{W}$, which contain row $i_{m}$ of $\mathbf{W}$ on row $i_{m}$ and zeros elsewhere (see Equation (\ref{eq:QiDef})).
\end{theorem}
\begin{proof}
This can be shown by extending the proof of Theorem~\ref{thm:MultiEdgeLikelihood} similarly to the extension from the one-edge likelihood in Theorem~\ref{thm:OneEdgeLikelihood} to the one-node likelihood in Theorem~\ref{thm:OneNodeLikelihood}.
However, the presentation is omitted from here for brevity.
\end{proof}

The next section is devoted to verifying the above result with an artificial experiment, and in Section~\ref{sec:ApplicationToWildAnimalMovement} its functionality and usability is tested with real data of incomplete trajectories of wild reindeer.

\subsection{Validation of $\betaMLE$ with incomplete trajectories}
\label{sec:MultiObsValidation}
We evaluated the applicability and precision of the MLE method in the case of incomplete trajectories from $s$ to $t$. 
We only considered incomplete node trajectories and the corresponding method for computing the MLE. 
The setting was similar to the experiment for validating the MLEs for complete trajectories in Section~\ref{sec:CompleteTrajectories}.
Namely, we again generated 200 paths for different values of $\beta$ between uniformly distributed $s$-$t$-pairs, where $s$ and $t$ were at least 3 steps apart from each other.
We used the same three graphs that were used in Section~\ref{sec:CompleteTrajectories}, i.e., a grid uniform costs, a simulated Gaussian landscape and a graph generated with the LFR algorithm.

From the 200 generated paths, we extract a set $\Omega$ of 200 incomplete node trajectories by sampling $M$ nodes, where $M = \min(300, M')$, and $M'$ is drawn uniformly from $1,\ldots,L(\wp)$.
We limit the maximum length of an incomplete trajectory to 300 for computational efficiency.
We then compute $\hat{\beta}_{\mathrm{MLE}}$ by performing a line search on $\beta$ for finding the maximum value of $\logL (\beta \giv \Omega)$.
This is again repeated 10 times and we report the mean and standard deviation of the 10 obtained MLE values.

The results are gathered in Table \ref{tab:IncompletePaths}.
They show that the MLEs are accurate throughout the tested range of values of $\beta$ and thus confirm that the methodology derived in Section \ref{sec:IncompleteTrajectories} can be used in practice.
Moreover, quite surprisingly, the results are almost as accurate as the results in Table~\ref{tab:MultipleFullPaths} in the experiment with complete trajectories of Section~\ref{sec:ValidationOfFullPathMLE}.

\begin{table}
\centering
{\footnotesize 
\begin{tabular}{r|c c c}
$\beta$
& $\hat{\beta}_{\mathrm{MLE}}$, uniform grid
& $\hat{\beta}_{\mathrm{MLE}}$, simulated landscape
& $\hat{\beta}_{\mathrm{MLE}}$, LFR graph
\\
\hline
$   0.001$ & $ 0.00101 \pm 0.00016$ & $ 0.00106 \pm 0.00013$ & $ 0.00108 \pm 0.00033$ \\
$   0.005$ & $ 0.00497 \pm 0.00043$ & $ 0.00510 \pm 0.00069$ & $ 0.00471 \pm 0.00119$ \\
$    0.01$ & $ 0.00980 \pm 0.00070$ & $ 0.00992 \pm 0.00088$ & $ 0.01045 \pm 0.00198$ \\
$    0.05$ & $ 0.05014 \pm 0.00275$ & $ 0.05091 \pm 0.00250$ & $ 0.05297 \pm 0.00860$ \\
$     0.1$ & $ 0.10117 \pm 0.00704$ & $ 0.09433 \pm 0.00807$ & $ 0.10384 \pm 0.01226$ \\
$     0.5$ & $ 0.50810 \pm 0.03167$ & $ 0.49411 \pm 0.02237$ & $ 0.49841 \pm 0.05781$ \\
$       1$ & $ 1.01074 \pm 0.07147$ & $ 0.98349 \pm 0.03991$ & $ 0.98186 \pm 0.07691$ \\
$       5$ & $ 4.92557 \pm 0.27878$ & $ 4.93601 \pm 0.24915$ & $ 5.09714 \pm 0.45950$ \\
$      10$ & $12.73153 \pm 5.07237$ & $10.08324 \pm 0.28333$ & $10.49840 \pm 0.98032$ \\
\end{tabular}

}
\caption{The MLEs in the experiment with incomplete node trajectories. The first column on the left shows the value of $\beta$ used for generating the paths and the two other columns show the mean $\pm$ the standard deviation of the MLEs over 10 repetitions.
}
\label{tab:IncompletePaths}
\end{table}

\subsection{Binomial distribution on number of observations}
\label{sec:Binomial}
So far, we assumed that the number of observations in an incomplete trajectory is a random variable distributed uniformly over the range $[1, L(\wp)]$, where $\wp$ is the trajectory being observed.
In this section we briefly discuss an alternative for the uniformity assumption by considering a binomial distribution over the number of observations, $\mu \sim \mathrm{Bin}(L(\wp), p_{\mu})$ with the condition that $L(\wp)>0$.
We only derive here a way for computing the likelihood from incomplete trajectories with this assumption, but leave further examination of this idea for future work.

Considering that the random variable corresponding to the number of observations of an incomplete trajectory, $\mu$, is binomially distributed, involves an additional parameter, $p_{\mu}$, i.e., the probability of making an observation at each step of an observed path. Thus, using this assumption means that $p_{\mu}$ can be estimated from the data by some means.
Assuming a binomial distribution for $\mu$,
the probability of observing $M \geq 1$ edges from a path $\wp$ is
\begin{align}
\begin{aligned}
\ProbAll(\mu = M \giv \mu \geq 1) 
&=
\dfrac{\ProbAll(\mu = M)}{1-\ProbAll(\mu = 0)}
\\
&=
\dfrac{1}{1-q_{\mu}^{M}}
\binom{L(\wp)}{M}
p_{\mu}^{M}
q_{\mu}^{L(\wp)-M},
\\
&=
\dfrac{
p_{\mu}^{M}
q_{\mu}^{L(\wp)}
}{
q_{\mu}^{M}
(
1-q_{\mu}^{M}
)
}
\binom{L(\wp)}{M}
,
\end{aligned}
\end{align}
where $p_{\mu}$ is a fixed probability that an observation is recorded at any given step of path of $\wp$ and $q_{\mu} = 1-p_{\mu}$. 
Inserting this to the expression of the likelihood in Equation~(\ref{eq:MultiEdgeLikelihood2}) gives
\begin{align}
\Lhood
(\beta \giv \tilde{e})
&=
\dfrac{
p_{\mu}^{M}
}{
q_{\mu}^{M}
(
1-q_{\mu}^{M}
)
}
\psum
\Prob(\wp)
\V_{\tilde{e}}(\wp)
q_{\mu}^{L(\wp)},
\label{eq:LikelihoodBinomialAssumption}
\end{align}
which can be computed by replacing the definition of $\mathbf{L}_{\tilde{e}}$, instead of Equation~(\ref{eq:OmegaMatrixLogarithm}), as
\begin{equation}
\mathbf{L}_{\tilde{e}} 
= 
\sum_{k=1}^{\infty} 
(
q_{\mu}
\mathbf{Q}_{\tilde{e}}
)^k
=
\left(
\mathbf{I}
-
q_{\mu}
\mathbf{Q}_{\tilde{e}}
\right)^{-1}
\label{eq:BinomialOmegaMatrixLogarithm}
\end{equation}
We only need the element 
$
\mleft[ 
\mathbf{L}_{\tilde{e}} 
\mright]_{s,Mn + t}
$
from this matrix, which can be obtained, for instance,~by solving the $(Mn+t)$-th column,
$
\mathbf{l}_{Mn+t}^{\tilde{e}}
$,
of 
$
\mathbf{L}_{\tilde{e}}
$ 
from the linear system
\begin{equation}
\label{eq:BinomialInverse}
\left(
\mathbf{I}
-
q_{\mu}
\mathbf{Q}_{\tilde{e}}
\right)
\mathbf{l}_{Mn+t}^{\tilde{e}}
=
\mathbf{e}_{Mn+t}
\end{equation}
and picking its $s$-th element.
Then, the likelihood is given by replacing Equation~(\ref{eq:MultipleEdgeFullLhood}) with
\begin{equation}
\label{eq:BinomialLikelihoodFinal}
\Lhood
(\beta \giv \tilde{e})
=
\dfrac{
p_{\mu}^{M}
}{
q_{\mu}^{M}
(
1-q_{\mu}^{M}
)
}
\cdot
\dfrac{
\mleft[ 
\mathbf{L}_{\tilde{e}} 
\mright]_{s,Mn + t}
}{\PF}.
\end{equation}

Computing 
$
\mleft[ 
\mathbf{L}_{\tilde{e}} 
\mright]_{s,Mn + t}
$
from the matrix inverse in Equation~(\ref{eq:BinomialInverse}) can be done more efficiently than the computation from Equation~(\ref{eq:OmegaMatrixLogarithm}), which cannot be expressed in closed form.
However, as mentioned earlier, this method requires setting the observation probability $p_{\mu}$ as an additional parameter, which is not needed in the method based on the uniform distribution.
However, as already mentioned, in the remainder of the paper, we will only focus on the uniformity assumption, as described in Section~\ref{sec:MultipleEdgeLikelihood}, and leave the further study of the binomial assumption, discussed here, for future work.

\section{Application to animal movement modelling}
\label{sec:ApplicationToWildAnimalMovement}

\begin{figure}
	\centering
	\subcaptionbox{\label{fig:ReindeerLandscape}}{\includegraphics[width=.4\textwidth, trim={0 2cm 0 1.5cm}, clip]{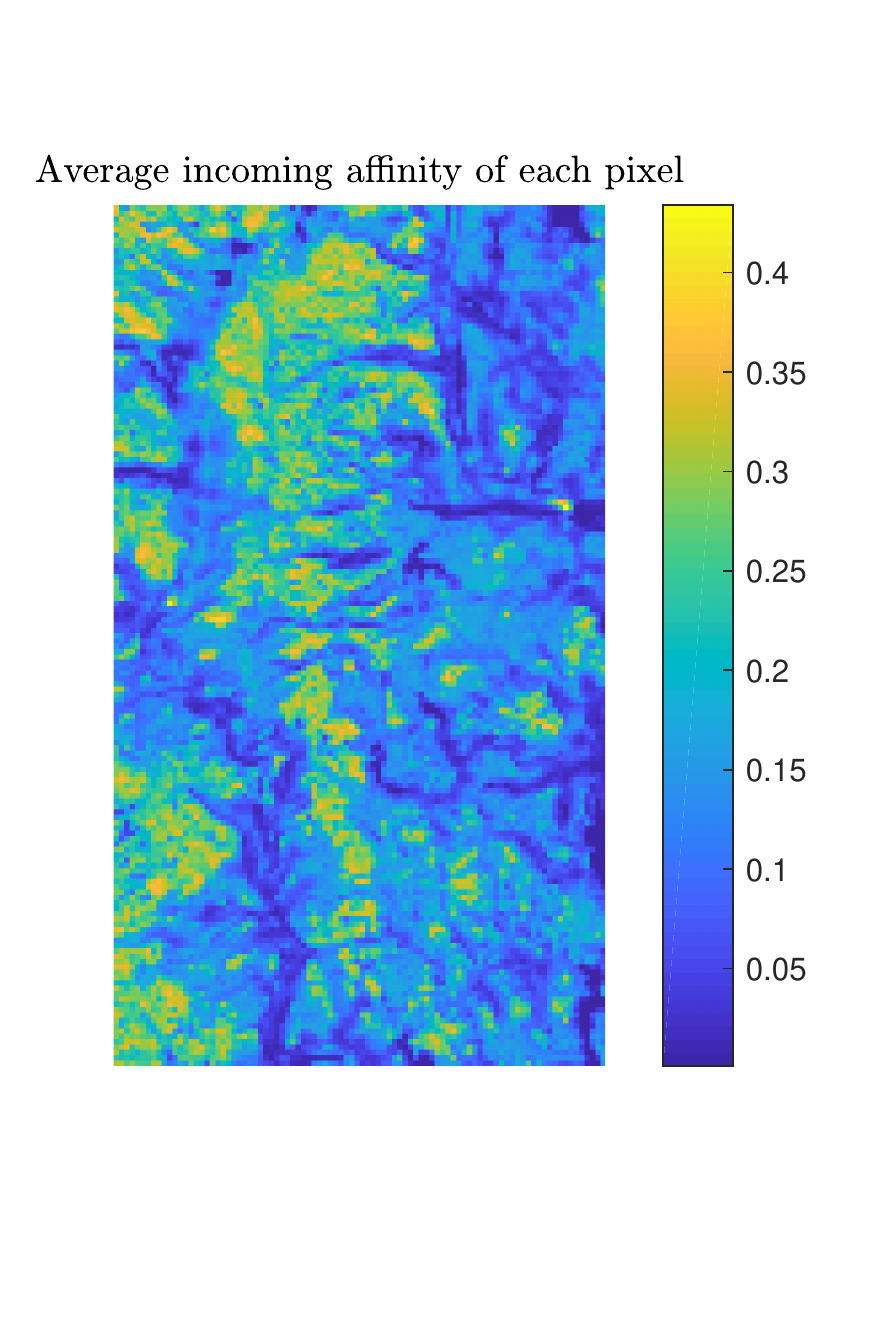}
	}
	\subcaptionbox{\label{fig:ReindeerVisits}}{\includegraphics[width=.4\textwidth, trim={0 2cm 0 1.5cm}, clip]{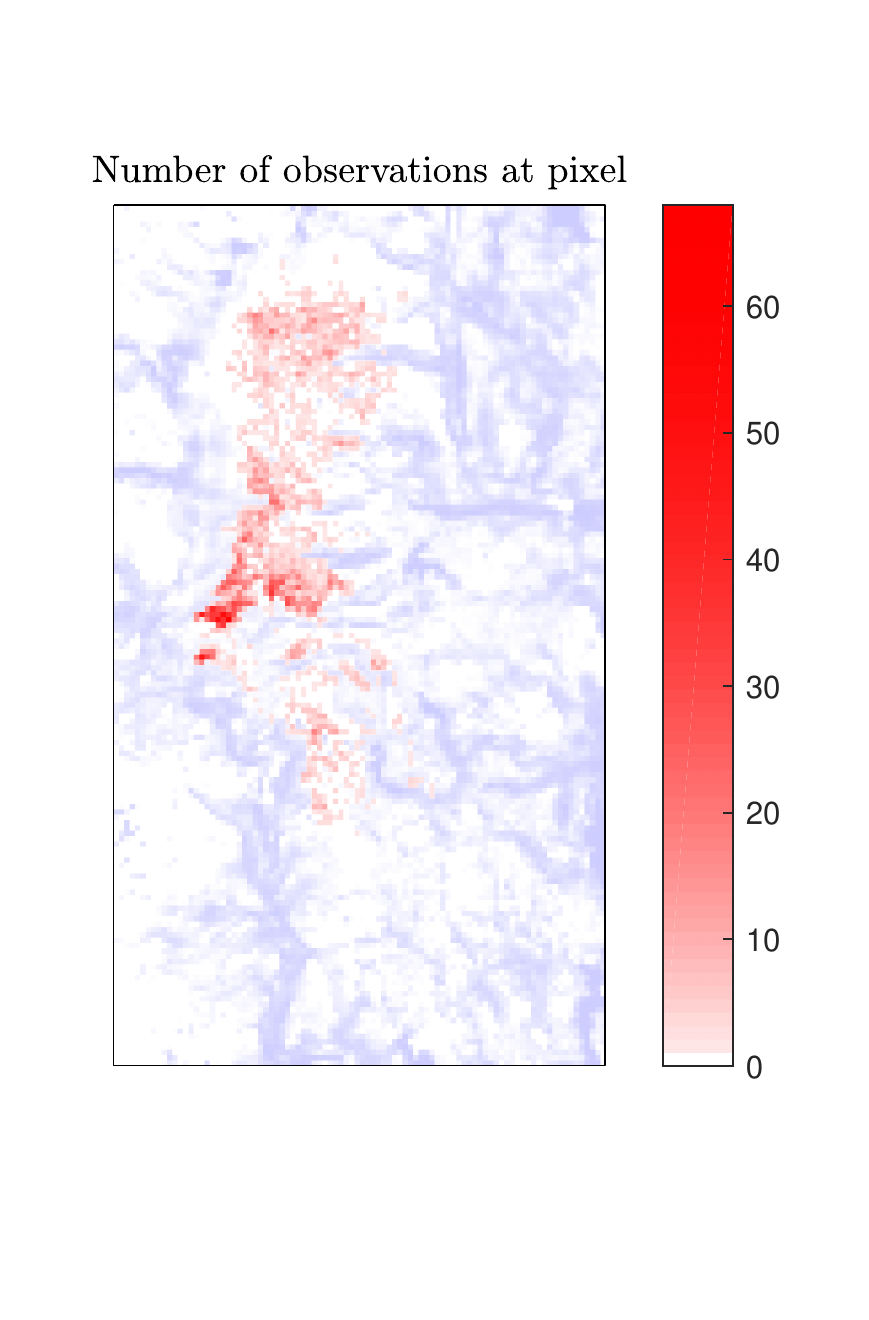}
	}
	\caption{Visualization of the landscape used for studying wild reindeer movement showing the average incoming edge affinity of each pixel as a heatmap~(A). The heatmap of the number of times an animal was observed at each pixel (in red hue; the blue hue shows the low-affinity (i.e.\ high cost) pixels simply in order to display the shape of the landscape on the image)  based on the GPS data (B).}
	\label{fig:LanscapeAndVisits}
\end{figure}

In order to test the applicability of the method for computing the MLE of $\beta$ in a real data setting, we ran the MLE method for a set of GPS trajectory data collected from individual wild mountain reindeer in the Austhei area of Norway.
The reason for studying this data is based on the apparent suitability of RSPs for modelling animal movement.
This holds especially for the study of wild reindeer, which are a migratory species, that often have a good sense of knowledge of their environment.
Accordingly, they can be considered to follow fairly optimal routes when moving in a landscape.
It would be, however, unrealistic to assume that the animals always follow only the optimal path on a static landscape.
Instead their movement decisions may be considered to involve some randomness.
This is exactly the kind of scenario that the RSP framework is designed for.
However, the aim of the experiments reported here was simply to test whether the MLEs are practical to compute for real data with the methods derived above and to see whether the estimates obtained seem sensible.
We leave a more careful analysis of the obtained results, and their applicability in more focused ecological problems, such as actual prediction of movement, for future work.

The data considered here consists of incomplete trajectories recorded during springtime migration, when the reindeer traverse the landscape from north to south crossing over a road passage cutting through the landscape.
The same area and partly the same data was studied previously in the context of RSPs in~\citep{panzacchi2016predicting}, although there the graph was constructed differently compared to the experiment presented here.
In~\citep{panzacchi2016predicting} a method was also devised for estimating the inverse temperature parameter, however based on a very different, more situation-dependent approach, compared to the more generic method introduced in this work.

\begin{figure}
\centering
\subcaptionbox{
\label{fig:Trajectory1}}{
\makebox[0.33\linewidth][c]{
\includegraphics[width=0.23\linewidth]{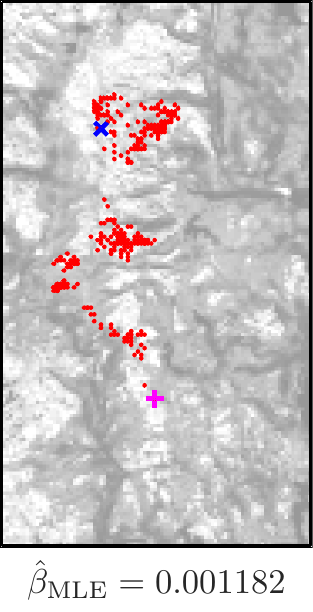}%
}%
}%
\subcaptionbox{
\label{fig:Trajectory2}}{
\makebox[0.33\linewidth][c]{
\includegraphics[width=0.23\linewidth]{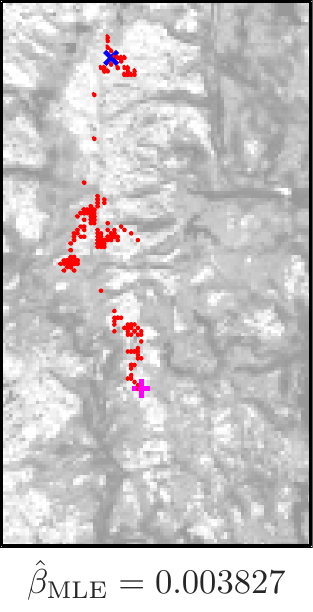}%
}%
}%
\subcaptionbox{
\label{fig:Trajectory3}}{
\makebox[0.33\linewidth][c]{
\includegraphics[width=0.23\linewidth]{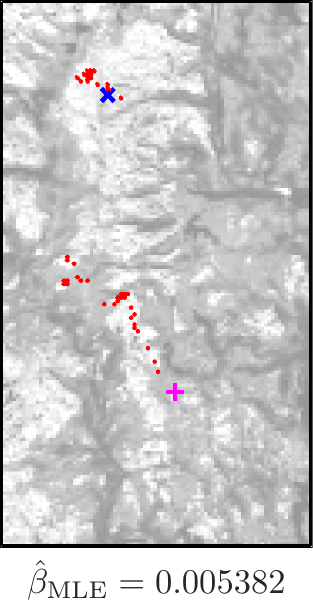}%
}%
}%
\caption{Examples of three trajectories of individual reindeer and the MLE of $\beta$ of each trajectory. The starting point, on the north side of the landscape, is depicted with a blue '\texttt{x}', and the ending point with a magenta '\texttt{+}'. All 32 trajectories in the data have different source and target nodes $s$ and $t$.}
\label{fig:Trajectories}
\end{figure}

The landscape is modeled as a rectangular grid graph, where each node represents a 500 meter wide square pixel of a map image of the landscape.
The size of the landscape is $80.5$ km $\times$ $46$ km, resulting in a graph with $14 812$ nodes.
As before, each pixel was again connected to its 8 surrounding pixels.
First, the edge affinities $a_{ij}$ were defined as Step Selection Probabilities (SSP) \citep{lele2006weighted}, inferred by detecting features on each pixel with remote sensing methods and by measuring the preference of those features for movement based on the reindeer GPS data.
A more detailed explanation of how the affinities were defined is provided in Appendix~\ref{app:SSPF}.
The edge costs were defined simply as the reciprocals of affinities, i.e.~$c_{ij} = 1/a_{ij}$.
The landscape is illustrated in Figure \ref{fig:ReindeerLandscape} as a heatmap where the pixels are colored according to the average incoming edge affinity of each pixel.

The data contains 32 incomplete trajectories recorded during the years 2007-2013 from 14 different individuals.
Figure \ref{fig:ReindeerVisits} visualizes the trajectory data by marking the number of times any of the 32 trajectories were observed at each pixel.
The GPS measurements were made for the most part every 3 hours, 50 \% of trajectories contained all locations. 
However, many trajectories contain gaps of 6 hours to one day, and one trajectory is missing locations for nearly two weeks. 
Trajectories started in the winter range between March 1st and 20th, and ended in the summer range between June 8th and 30th. 
We cut the original trajectories after they have crossed a certain line which can be interpreted as a border of the summer range.
The number of observations in the resulting trajectories was between 57 and 535.
For each observation, we detect the pixel that the animal is in, and use the method derived in Section \ref{sec:IncompleteTrajectories} for computing the multiple-node likelihood.
We consider the trajectories as extracted from hitting paths, where the hitting node is selected to be the first pixel within the summer range where the animal has been observed.
As discussed in Section~\ref{sec:Determining_s_and_t}, more sophisticated methods could be used for determining the hitting target node, such as the \emph{margin-constrained RSP model}~\cite{guex2019randomized}, where the graph nodes are assigned probabilities of being a starting or ending node of a path.

\begin{figure}
\centering
\includegraphics[width=0.6\textwidth]{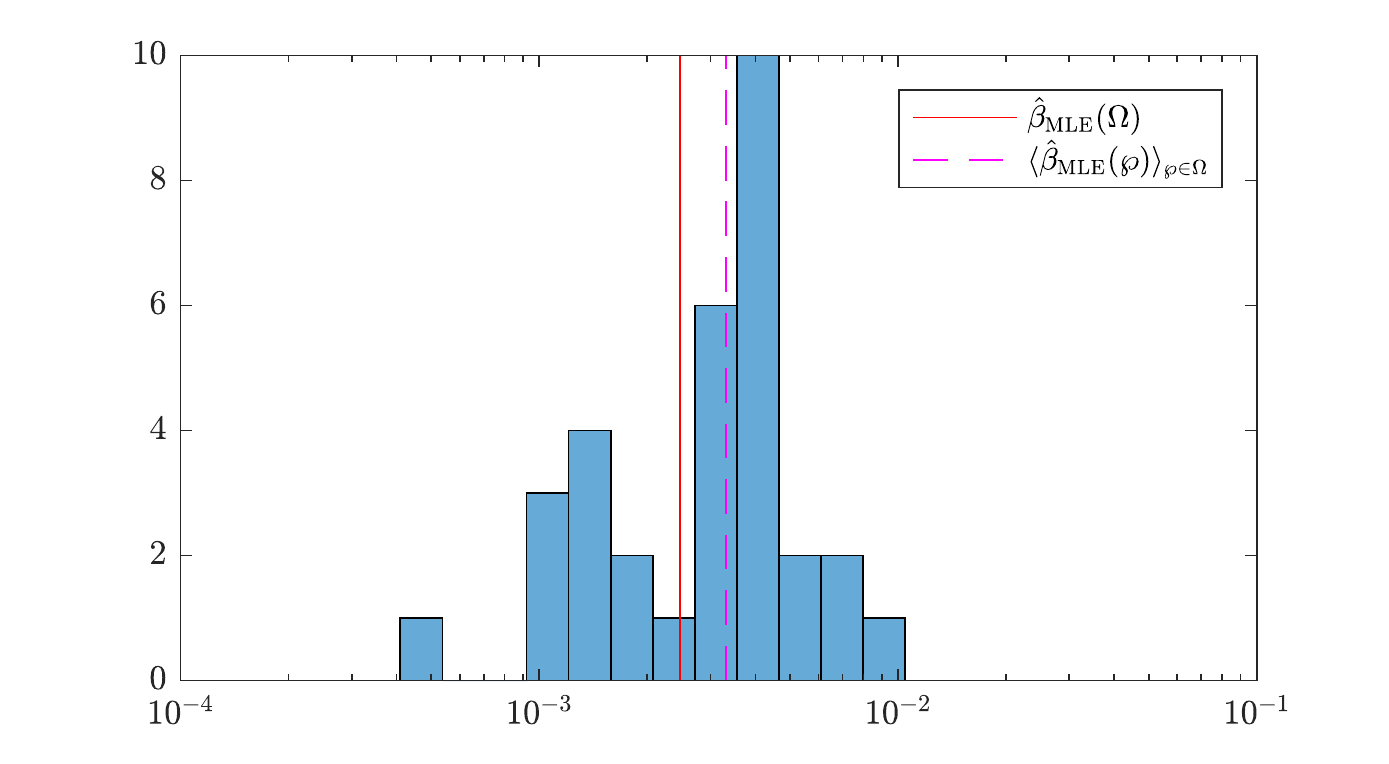}
\caption{Histogram of the MLE values for the 32 trajectories of reindeer. The solid red line marks the MLE obtained with the whole set of trajectories, $\betaMLE (\Omega) = 0.002466$, whereas the magenta dashed line shows the mean of the individual trajectory MLEs, $\langle \betaMLE (\wp) \rangle_{\wp \in \Omega} = 0.003311$.}
\label{fig:MLEHistogram}
\end{figure}

We first computed $\betaMLE$ for each trajectory separately, resulting in 32 MLE values. 
Figure~\ref{fig:Trajectories} shows three examples of trajectories and the corresponding  $\betaMLE$ values.
As can be seen, the MLE values from Figure~\ref{fig:Trajectory1} to Figure~\ref{fig:Trajectory3} increase.
Similarly, the directness of the trajectories, in that order, seemingly increases.
The trajectory in Figure~\ref{fig:Trajectory3} appears to circulate fairly randomly before arriving at the destination, whereas the trajectory in Figure~\ref{fig:Trajectory3} is heads quite straightforwardly towards the destination.
The trajectory in Figure~\ref{fig:Trajectory2} appears as an intermediate type between the two others.
These examples provide a sanity check and show that the MLE method gives reasonable and meaningful estimates of $\beta$.

A histogram of the MLE values of the 32 trajectories is presented in Figure~\ref{fig:MLEHistogram}.
The MLEs remained for the most part in a fairly consistent range of values, between $[ 0.001, 0.01 ]$, with one exceptional trajectory obtaining a MLE below this range, namely $\betaMLE= 0.0004715$.
The mean of the individual MLEs was $\langle \betaMLE (\wp) \rangle_{\wp \in \Omega} = 0.003311$, and is represented in Figure~\ref{fig:MLEHistogram} as the dashed magenta line.

In addition to estimating the individual trajectory MLEs, we also considered the whole set of trajectories as being generated by the same value of $\beta$, and thus computed the population-wide MLE value using the whole set of trajectories.
Note that all the trajectories are between different $s$-$t$-pairs, where $s$ is simply the first observed node of the trajectory, on the winter range, and $t$ is, as explained earlier, the first observed node on the trajectory that is located on the summer range.
We consider the likelihood of the set as the product of the likelihoods of each trajectory, i.e.\ that the trajectories are independent.
Assuming independence between the trajectories is convenient for computation, but is also justified by the fact that the trajectories are mostly collected over different years and can be from individuals belonging to different herds.
The MLE given by the whole set of trajectories was $\betaMLE(\Omega) = 0.002466$, and is marked by the solid red line in Figure~\ref{fig:MLEHistogram}.

\begin{figure}
	\centering
	\subcaptionbox{\label{fig:RSPExpectedVisits1}}{
	\includegraphics[width=.3\textwidth, trim={1cm 2cm 0.7cm 1.5cm}, clip]{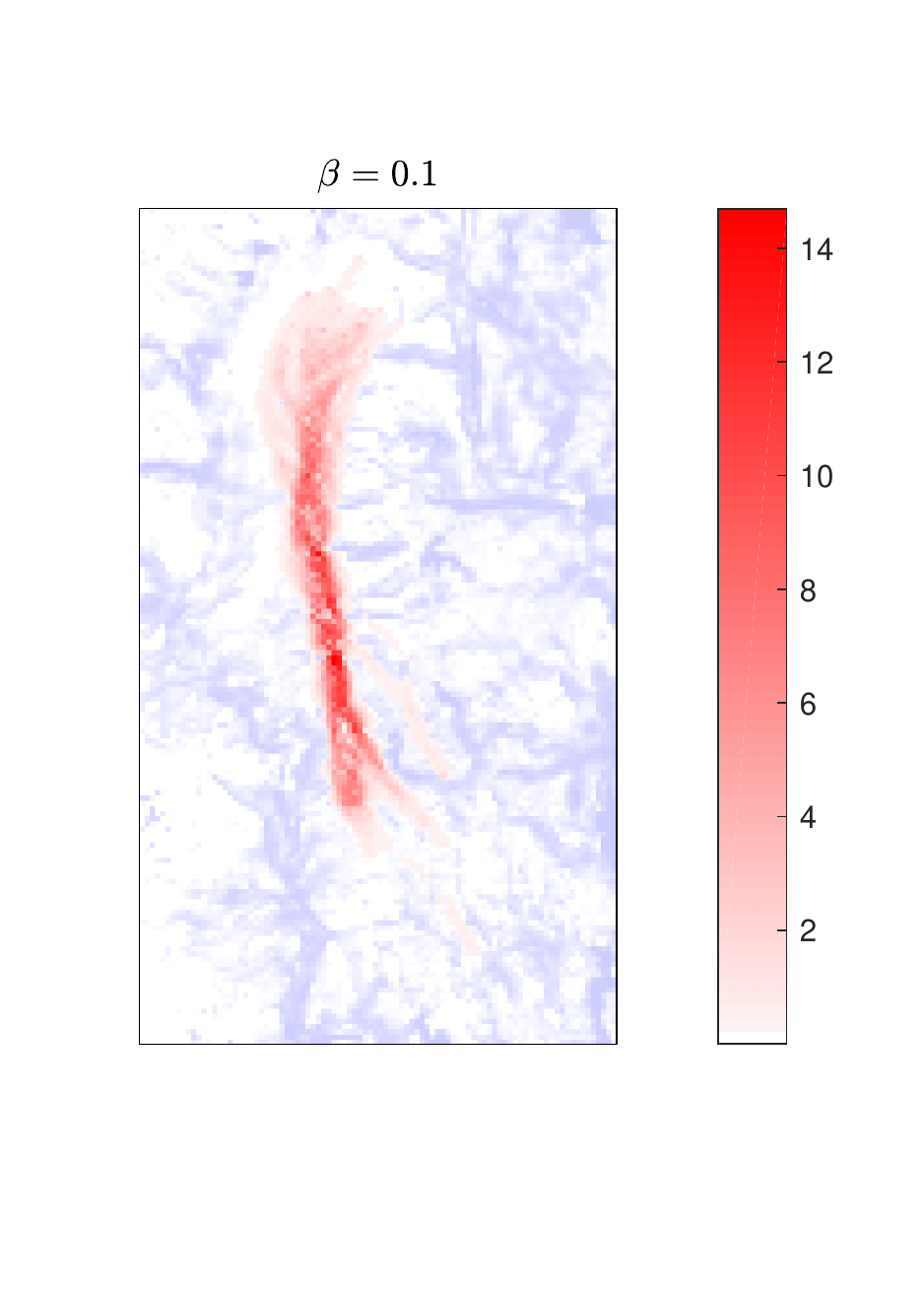}
	}
	\subcaptionbox{\label{fig:RSPExpectedVisits2}}{
	\includegraphics[width=.3\textwidth, trim={1cm 2cm 0.7cm 1.5cm}, clip]{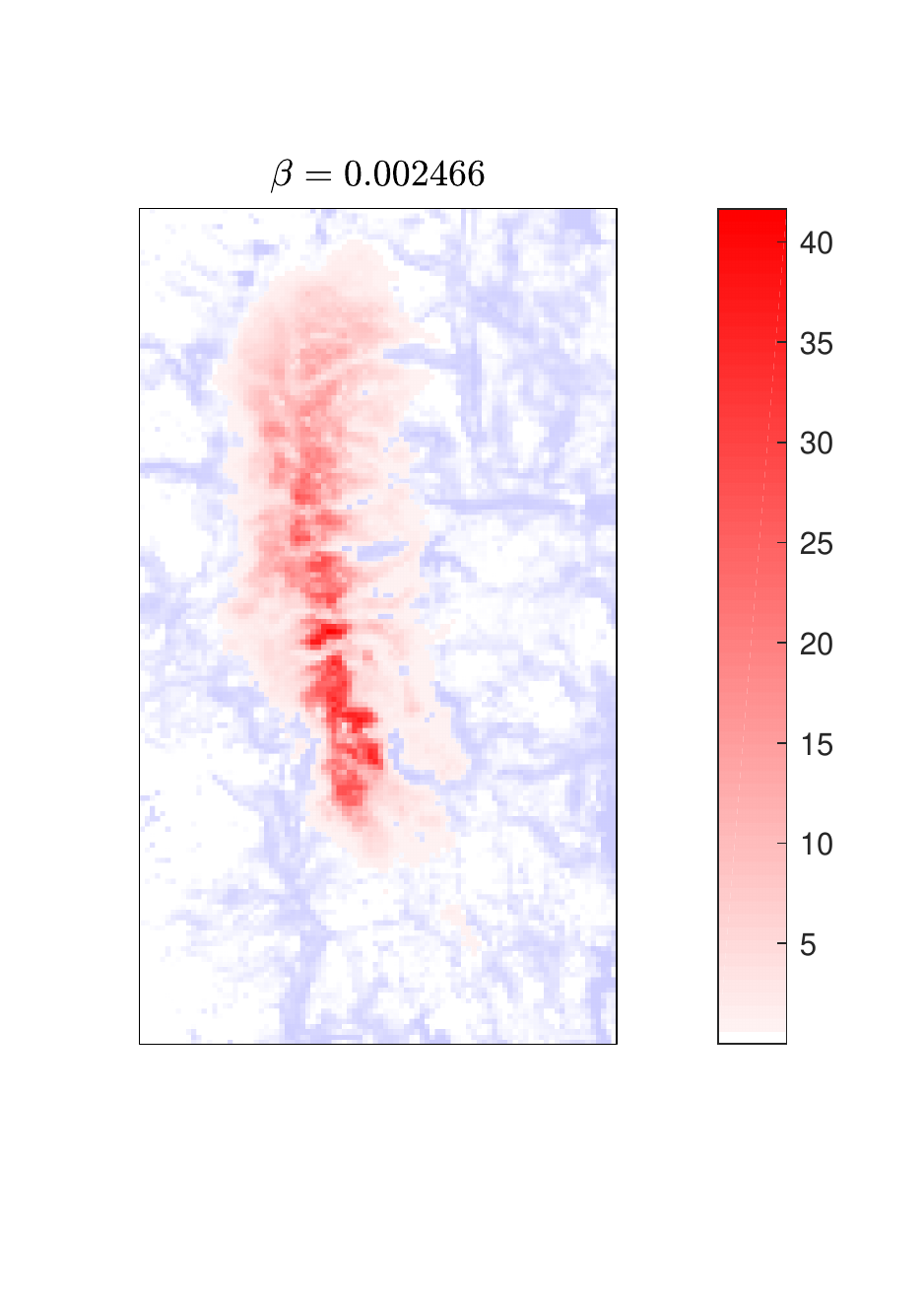}
	}
	\subcaptionbox{\label{fig:RSPExpectedVisits3}}{
	\includegraphics[width=.3\textwidth, trim={1cm 2cm 0.7cm 1.5cm}, clip]{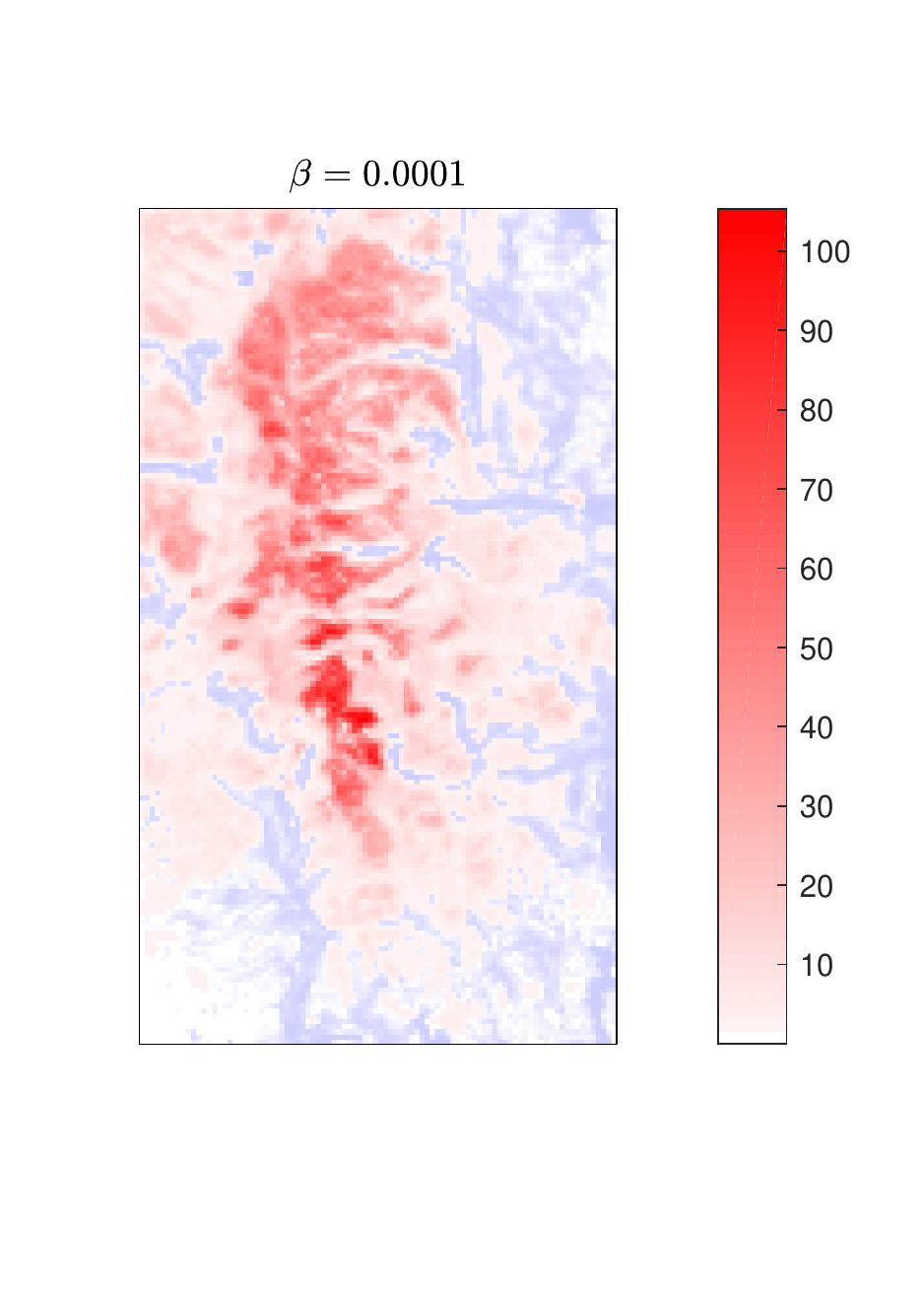}
	}
	\caption{The expected number of visits to each pixel over the RSP probability distributions between each observed $s$-$t$-pair, with three values of $\beta = 0.1$ (a), $\betaMLE(\Omega) = 0.002466$ (b; the MLE based on the set of trajectories $\Omega$), and $\beta=0.0001$ (c). The red hue at pixel $i$ marks the value $\sum_{(s,t) \in X} \expVi(s,t)$, given the value of $\beta$ depicted above the plot. Here, $X$ denotes the set of $s$-$t$-pairs of the trajectories in the data set $\Omega$.}
	\label{fig:RSPExpectedVisits}
\end{figure}

Figure \ref{fig:RSPExpectedVisits} contains plots for the expected numbers of visits according to the RSP model with different values of $\beta$.
These were computed by considering the starting and ending nodes, $s$ and $t$, of each trajectory in the data separately, computing the expected number of visits to each pixel for that $s$-$t$-pair, and by summing the contributions of each $s$-$t$-pair.
In each plot, the blue hue represents the average incoming edge cost of each pixel, and the transparent red hue represents the observed or expected number of visits by an animal to a pixel.

Figure \ref{fig:RSPExpectedVisits1} corresponds to the model computed with  $\beta=0.1$, and Figure \ref{fig:RSPExpectedVisits3} with $\beta=0.0001$ for each $s$-$t$-pair.
In between, Figure \ref{fig:RSPExpectedVisits2} shows the result with the value $\beta = 0.002466$, which is the MLE estimated by using the whole set of trajectories.
We also constructed another plot by using for each $s$-$t$-pair the MLE of $\beta$ of the corresponding trajectory, but it is not presented here, because it resembles almost exactly the plot obtained with the global MLE in Figure \ref{fig:RSPExpectedVisits2}.

A visual inspection of the different plots obtained with the RSP distribution using different values of $\beta$ indicates that the model using the MLE values of $\beta$ has most resemblance with the plot of the actual observations.
Note, however, that the plots in Figure \ref{fig:RSPExpectedVisits} are not exactly expected to resemble the plot in Figure \ref{fig:ReindeerVisits}.
Namely, as explained before, the RSP model does not take into account the temporal aspect of the movement.
The temporal aspect, however, is grained in the trajectories, as the speed of movement and the time interval between observations vary.
For instance, after crossing the road over to the south side, the animals tend to move faster away from the road \citep{panzacchi2013road}. 
This results generally in a fewer number of observations in the pixels south of the road.
In other words, Figure \ref{fig:ReindeerVisits} rather depicts the \emph{time spent} at each pixel by the animals.
Instead, the plots in Figure \ref{fig:RSPExpectedVisits}, given by the RSP model, reflect the importance of each pixel as an intermediate point of movement, which can be more crucial, for instance, for detecting corridors and barriers in a landscape.

\section{Conclusion}
\label{sec:MLE-Discussion}

This paper focused on the estimation of parameters when fitting the RSP model to data containing trajectories on a network, with most focus on the estimation of the inverse temperature parameter $\beta$.
Methods were derived for computing and maximizing the likelihood of values of $\beta$ given a data set of either fully or only partly observed trajectories~(in Sections~\ref{sec:CompleteTrajectories} and \ref{sec:IncompleteTrajectories}, respectively).
The methods were shown to provide accurate estimates of $\beta$ with simulation examples where trajectories were generated on artificially constructed graphs.
In addition to estimation of $\beta$, also maximum likelihood estimates of the edge costs, $c_{ij}, (i,j) \in E$, of a graph were derived, when dealing with complete trajectories.

The MLE method derived for incomplete trajectories was also tested on real data of trajectories of wild mountain reindeer on a landscape area in Norway.
The purpose of the experiment was only to verify further that the MLEs of $\beta$ can be computed in practice and that the estimates seem reasonable and sensible.
This indicates that the MLE method provides a well-founded and functional way for model fitting, when using the RSP model for more specific applications related to animal movement.

The most significant theoretical results of the paper, in Section~\ref{sec:IncompleteTrajectories}, dealt with fitting the parameter $\beta$ to data consisting of incomplete trajectories.
As a computational peculiarity, the derivation of the likelihood function in the special case of observing only one edge or node of a trajectory was shown to involve the matrix logarithm.
In fact, the computational techniques appearing in the derivation could be applied for various other network analysis purposes.
For instance, a similar method can be derived for computing the expected average edge cost of paths, which can be of use in situations where, for instance, two nodes may be considered more reliably connected when there are, on average, no high-cost edges between them. 
Also, as discussed in Section~\ref{sec:IncompleteTrajectories}, the matrix logarithm could be used for computing the probability of observing an edge or node over the natural random walk distribution between an $s$-$t$-pair, assuming the uniform distribution over the path for the sampling process.
This could be developed further to define a new network centrality measure based on the observation probability of edges or nodes.

One further extension of the theory developed in this work is to consider features on the edges or nodes of the network as separate costs, and to fit parameters on those features assuming the RSP model.
Likelihood maximization could also be developed for other observation scenarios than the one dealt in this work.
One example would be a case where some of the nodes or edges of the network contain sensors and the trajectories can only be observed when they visit these sensor nodes or edges.
Also, the RSP framework will in future work be extended to temporal networks, which will bring more challenges also to the parameter estimation problem.
Lastly, the methods developed here will hopefully indicate new principled ways of selecting an appropriate value for $\beta$ for other network data analysis tasks that the RSP framework is used for.

\newpage
\section*{Appendix: Additional material and proofs of main results}
\label{TheAppendix}

\appendix
\section{Proof of Theorem~\ref{thm:OneNodeLikelihood}}
\label{app:OneNodeLikelihood}
As before, let $\rho$ be the random variable corresponding to the drawing of a path from the set $\Pset$ according to the RSP distribution (Equation (\ref{eq:GibbsBoltzmann})).
Given a path 
$\wp = (\wp(0)=s, \wp(1), \ldots, \wp(L(\wp))=t)$, let us consider two random variables: 
\begin{itemize}
\item 
$\nu_{0}$, corresponding to the drawing of an intermediate node, drawn uniformly from the subsequence of $\wp$ excluding the last node, but \emph{including the first node}, i.e.,~%
$(\wp(0), \ldots, \wp(L(\wp) - 1))$, and 
\item
$\nu_{1}$ corresponding to the drawing of an intermediate node, drawn uniformly from the subsequence of $\wp$ \emph{excluding both the first and last node}, i.e.,~%
$(\wp(1), \ldots, \wp(L(\wp) - 1))$.
\end{itemize}
Recalling the assumptions stated in Section~\ref{sec:OneNodeLikelihood}, the one-node likelihood is given by the probability distribution of $\nu_{1}$.
However, we first derive the distribution of $\nu_{0}$ as an intermediate result and use it to derive the distribution of $\nu_{1}$.

Using $\Vi(\wp) = \sum_{j \in Succ(i)} \Vij(\wp)$, as earlier, the conditional probability of node $i$ as the outcome of $\nu_{0}$, given a path $\wp \in \Pset$, is
\begin{equation}
\mathrm{P}
(
\nu_{0} = i 
\giv
\rho = \wp
)
=
\dfrac{\expVi (\wp)}{L(\wp)}
=
\dfrac{\sum_{j\in Succ(i)} \expVij (\wp)}{L(\wp)}.
\label{eq:Nu0ConditionalProbability}
\end{equation}

Then, marginalizing out $\rho$, we can write the distribution of $\nu_{0}$, in similar fashion to Equation~(\ref{eq:OneEdgeLikelihoodLast}) for the one-edge likelihood, as
\begin{align}
\label{eq:Nu0First}
\mathrm{P}
(\nu_{0} = i ; \beta)
&=
\dfrac{1}{\PF}
\sum_{j \in Succ(i)}
\sum_{\wp \in \Pset} 
\dfrac{\tw(\wp) \Vij(\wp)}{L(\wp)}
\\
&= 
-\dfrac{1}{\beta \PF} 
\sum_{j\in Succ(i)} 
\dfrac{\partial \widetilde{\mathcal{Z}}_{st}}{\partial c_{ij}} 
.
\label{eq:OneNodeLikelihood}
\end{align}
This equation is the equivalent of Equation (\ref{eq:OneEdgeLikelihoodLast}) for edges.
Now, similar to Equation (\ref{eq:WkDerivative}), we have, for any $k\geq 1$
\begin{align}
\sum_{j \in Succ(i)} 
\dfrac{\partial \Wt^{k}}{\partial c_{ij}} 
&= 
-\beta \!\! 
\sum_{j \in Succ(i)} 
\sum_{l=1}^{k} 
\Wt^{l-1} 
\mathbf{W}_{ij} 
\mathbf{W}^{k-l} 
\\
&= 
-\beta 
\sum_{l=1}^{k} 
\Wt^{l-1} 
\mleft(
\sum_{j \in Succ(i)} 
\mathbf{W}_{ij} 
\mright) 
\mathbf{W}^{k-l} 
\\
&= 
-\beta 
\sum_{l=1}^{k} 
\Wt^{l-1} 
\mathbf{W}_{\mathrm{r}(i)} 
\Wt^{k-l}
\\
&\triangleq 
-\beta 
\mathbf{S}_{i}^{(k)},
\label{eq:DerivativesOfWk}
\end{align}
where $\mathbf{W}_{\mathrm{r}(i)} = \mathbf{I}_{ii} \mathbf{W}$
is the matrix containing the $i$-th row of matrix $\mathbf{W}$ on its $i$-th row, but zeros elsewhere, and $\mathbf{S}_{i}^{(k)}$ defined analogously to $\mathbf{S}_{ij}^{(k)}$ in Equation (\ref{eq:SijkDef}).

The rest of the derivation of the distribution of $\nu_{0}$ proceeds as in the single-edge likelihood case, but with matrix $\mathbf{W}_{\mathrm{r}(i)}$ in place of matrix $\mathbf{W}_{ij}$.
Namely, the derivative expression appearing in Equation (\ref{eq:OneNodeLikelihood}) can be computed (as in Equations (\ref{eq:PFDerivativeFirst})-(\ref{eq:PFDerivative})), as
\begin{equation}
\sum_{j\in Succ(i)} 
\dfrac{\partial \widetilde{\mathcal{Z}}_{st}}{\partial c_{ij}} 
=
\mathbf{e}_{s}^{\mathsf{T}}
\sum_{k=1}^{\infty} 
\dfrac{\mathbf{S}_{i}^{(k)}}{k}
\mathbf{e}_{t}
\label{eq:NodeDerivativeOfTildeZ},
\end{equation}
which, again, can be computed as element $(s,n+t)$ of the matrix logarithm 
$
\mathbf{L}_{i} 
= 
-\mathbf{log}
(\mathbf{I}-\mathbf{Q}_{i})
$,
with the $(2n \times 2n)$ block matrix (see Equation (\ref{eq:PFDerivativeIntermediate}))
\begin{equation}
\label{eq:QiDefRecall}
\mathbf{Q}_{i} = 
\mleft[
\begin{array}{cc}
\Wt        & \mathbf{W}_{\mathrm{r}(i)} \\
\mathbf{0} & \Wt
\end{array}
\mright].
\end{equation}
In conclusion, we can write the distribution of $\nu_{0}$ as
\begin{equation}
\mathrm{P}
(\nu_{0} = i)
= 
-\dfrac{1}{\beta \PF} 
(-\beta [\mathbf{L}_{i}]_{s,n+t}) 
= 
\dfrac{[\mathbf{L}_{i}]_{s,n+t}}{\PF}.
\label{eq:Nu0DistributionFinal}
\end{equation}

Using the above, we can then derive the one-node likelihood according to the distribution of $\nu_{1}$.
We do this by considering a split of each path $\wp \in \Pset$ into a path consisting of the first step, from $s$ to one of its successor nodes $u \in Succ(s)$, and the rest of the path, from $u$ to $t$.
We denote, generally, by $\wp_{u} \in \Pset$ a path from $s$ to $t$ whose second node is a successor $u \in Succ(s)$ of $s$, and by $\wp_{u}' \in \mathcal{P}_{ut}$ the remainder of path  $\wp_{u}$ after the first step.
The likelihood is then given by
\begin{align}
\Lhood(\beta \giv i)
&=
\mathrm{P}(\nu_{1} = i; \beta)
=
\sum_{\wp \in \Pset}
\Prob (\wp)
\mathrm{P}
(\nu_{1}=i \giv \rho = \wp)
\\
&=
\dfrac{1}{\PF}
\psum
\widetilde{w}(\wp)
\mathrm{P}
(\nu_{1}=i \giv \rho = \wp)
\\
&=
\dfrac{1}{\PF}
\sum_{u \in Succ(s)}
\sum_{\substack{\wp_{u} \in \mathcal{P}_{st} \\ \wp_{u}(1) = u}}
\tw(\wp_{u})
\mathrm{P}
(\nu_{1}=i \giv \rho = \wp_{u})
\\
&=
\dfrac{1}{\PF}
\sum_{u \in Succ(s)}
w_{su}
\sum_{\wp_{u}' \in \mathcal{P}_{ut}}
\tw(\wp_{u}')
\mathrm{P}
(\nu_{0} = i \giv \rho = \wp_{u}')
\\
&=
\dfrac{1}{\PF}
\sum_{u \in Succ(s)}
w_{su}
\mathrm{P}(\nu_{0} = i)
\mathcal{Z}_{ut}
\\
&=
\dfrac{1}{\PF}
\sum_{u \in Succ(s)}
w_{su}
\mleft[
\mathbf{L}_{i}
\mright]_{u,n+t},
\label{eq:Nu1Distribution}
\end{align}
where
\begin{equation}
\mathrm{P}(\nu_{0} = i)
\mathcal{Z}_{ut}
=
\sum_{\wp_{u}' \in \mathcal{P}_{ut}}
w(\wp_{u}')
\mathrm{P}
(\nu_{0} = i \giv \rho = \wp_{u}')
\end{equation}
derives from 
\begin{equation}
\mathrm{P}(\nu_{0} = i)
=
\sum_{\wp_{u}' \in \mathcal{P}_{ut}}
\mathrm{P}
(\rho = \wp_{u}')
\mathrm{P}
(\nu_{0} = i \giv \rho = \wp_{u}')
\end{equation}
when considering paths from $u$ to $t$, instead of from $s$ to $t$.
This is the desired result, as expressed in Equation~(\ref{eq:OneNodeLikelihoodFinal}).

\qed

\section{Proof of Theorem~\ref{thm:MultiEdgeLikelihood} in the case of two observed edges}
\label{app:TwoEdgeLikelihood}
Here we describe a more rigorous derivation, compared to the heuristic justification of Section~\ref{sec:MultipleEdgeLikelihood}, for the computation of the two-edge likelihood, i.e.~the likelihood of observing two edges $e_{1}=(i_{1},j_{1}), e_{2} = (i_{2}, j_{2})$ along a path $\wp$.
The derivation of the likelihood for the case of an arbitrary number of observed edges, as expressed in Equation~(\ref{eq:MultipleEdgeFullLhood}) could be derived in a similar way.
However, here we deal only with the two-edge case for conciseness and clarity.

We begin by rewriting the form presented in Equation~(\ref{eq:MultipleEdgeLikelihood}) for the likelihood, with $M=2$:
\begin{equation}
\label{eq:TwoEdgeLikelihoodRecall}
\begin{aligned}
\Lhood(\beta \giv (e_{1}, e_{2})) 
&= \dfrac{1}{\PF} 
  \sum_{k=2}^{\infty} 
  \dfrac{1}{\binom{k}{2}} 
  \sum_{\wp_{k} \in \Pset^{(k)}} 
  \widetilde{w}(\wp_{k}) 
  \V_{e_{1} \leadsto e_{2}}(\wp_{k}),
\end{aligned}
\end{equation}
where now
$
\V_{e_{1} \leadsto e_{2}}(\wp_{k})
$
denotes the number of times that edges $e_1$ and $e_2$ appear on $\wp_k$ in that order, counting all occurences of $e_1$ and $e_2$ as separate.
For example, if
$
\wp = (1,2,3,2,3,4)
$, 
then
$
\V_{(1,2) \leadsto (2,3)}(\wp)
=
2,
$
as the edge $(2,3)$ appears twice on $\wp$ after the appearance of the edge $(1,2)$.

Now, $\V_{e_{1} \leadsto e_{2}}$ can be calculated by iterating over all edges of the path, checking if the edge corresponds to $e_{1}$, and then computing the number of times $e_{2}$ appears on the path after that.
Summing the occurences of $e_{2}$ after each occurence of $e_{1}$ then gives $\V_{e_{1} \leadsto e_{2}}$.
Formally, for any $\wp_{k} \in \Pset^{(k)}$ with $k\geq2$, where $\wp_{k} = (\wp(0), \wp(1), \ldots, \wp(k))$, we have
\begin{equation}
\label{eq:Etae1e2}
  \V_{e_{1} \leadsto e_{2}} (\wp_{k})
= \sum_{l=0}^{k-2} 
  \mleft[ 
  \wp_{k} \mleft(l,l+1 \mright) = e_{1} 
  \mright]
  \ \V_{e_{2}} (\wp_{k}(l+1:k)),
\end{equation}
where $\wp_{k}(l_{1}: l_{2})= (\wp(l_{1}), \ldots, \wp(l_{2}))$ denotes the subpath of $\wp_{k}$ from the $l_{1}$-th node to the $l_{2}$-th node, with $0 \leq l_{1} < l_{2} \leq k$, and the brackets $[ \ ]$ are the Iverson brackets~\citep{knuth1992two}, i.e.~1 if the statement within is true and 0 otherwise. 

Inserting Equation~(\ref{eq:Etae1e2}) into the sum appearing in Equation~(\ref{eq:TwoEdgeLikelihoodRecall}), we can write the sum in matrix form, for path length $k\geq 2$:
\begin{equation}
\begin{aligned}
\label{eq:SumOfwTimesEta}
&  \sum_{\wp_{k} \in \Pset^{(k)}}
   \widetilde{w}(\wp_{k}) 
   \V_{e_{1} \leadsto e_{2}}(\wp_{k})
   \\
= &
   \!\begin{aligned}[t]   
   \sum_{\wp_{k} \in \Pset^{(k)}}   
   \sum_{l=0}^{k-2}
  &w(\wp_{k}(0:l)) 
   \
   \mleft[ 
   \wp_{k} \mleft(l,l+1 \mright) = e_{1} 
   \mright] 
   \\
  &\times
   w_{i_{1}j_{1}} 
   w(\wp_{k}(l+1:k))
   \V_{e_{2}} (\wp_{k}(l+1 : k))
   \end{aligned}
   \\
= &
   \!\begin{aligned}[t]   
   \sum_{l=0}^{k-2}
   \sum_{\wp_{k} \in \Pset^{(k)}}   
   & w(\wp_{k}(0:l)) 
   \
   \mleft[
   \wp_{k}(l)=i_{1})
   \mright]
   \\
 & \times
   w_{i_{1}j_{1}} 
   \
   \mleft[
   \wp_{k}(l+1)=j_{1}
   \mright]
   \
   \mleft(
   -\dfrac{1}{\beta} 
   \dfrac{\partial}{\partial c_{i_{2}j_{2}}} 
   w(\wp_{k}(l+1:k))
   \mright)
   \end{aligned}
   \\
= &-\dfrac{1}{\beta}
   \sum_{l=0}^{k-2}
   \mathbf{e}_{s}^{\mathsf{T}} 
   \Wt^{l}\mathbf{e}_{i_{1}}
   w_{i_{1}j_{1}}
   \mathbf{e}_{j_{1}}^{\mathsf{T}} 
   \mleft(
   \dfrac{\partial \Wt^{k-l-1}}{\partial c_{i_{2}j_{2}}}   
   \mright)
   \mathbf{e}_{t}
   \\
= &\dfrac{1}{\beta^2}
   \sum_{l=0}^{k-2}
   \mathbf{e}_{s}^{\mathsf{T}}
   \Wt^{l}
   \mleft(
   \dfrac{\partial \Wt}{\partial c_{i_{1}j_{1}}}   
   \mright)
   \mleft(
   \dfrac{\partial \Wt^{k-l-1}}{\partial c_{i_{2}j_{2}}}   
   \mright)
   \mathbf{e}_{t}.
\end{aligned}
\end{equation}
But in fact, the matrix appearing above
\begin{equation}
\label{eq:Se1e2}
\mathbf{S}_{e_{1} \leadsto e_{2}}^{(k)}
=  \dfrac{1}{\beta^2}
   \sum_{l=0}^{k-2}
   \Wt^{l}
   \mleft[
   \dfrac{\partial \Wt}{\partial c_{i_{1}j_{1}}}   
   \mright]
   \mleft[
   \dfrac{\partial \Wt^{k-l-1}}{\partial c_{i_{2}j_{2}}}   
   \mright]
\end{equation}
can be computed, for any $k\geq 2$, as the $\mathbf{(1,3)}$-block of the $k$-th power of matrix 
\begin{equation}
\label{eq:TwoEdgesQijDef}
\mathbf{Q}_{e_{1}\leadsto e_{2}}
= \mleft[\begin{array}{ccc}
\Wt        & \mathbf{W}_{i_{1}j_{1}} & \mathbf{0}             \\
\mathbf{0} & \Wt                     & \mathbf{W}_{i_{2}j_{2}} \\
\mathbf{0} & \mathbf{0}              & \Wt
\end{array}\mright],
\end{equation}
which is the two-edge version of the matrix $\mathbf{Q}_{\tilde{e}}$ presented in Equation~(\ref{eq:QOmegaDef}).
For this, note that we can write
\begin{equation}
\label{eq:QijWithDerivatives}
\mathbf{Q}_{e_{1} \leadsto e_{2}}
= \mleft[\begin{array}{ccc}
\Wt        & \Wt_{i_{1}j_{1}} & \mathbf{0} \\
\mathbf{0} & \Wt             & \Wt_{i_{2}j_{2}} \\
\mathbf{0} & \mathbf{0}      & \Wt
\end{array}\mright]
=  \mleft[\begin{array}{ccc}
  \Wt
& -\dfrac{1}{\beta} \partial_{c_{i_{1}j_{1}}} \Wt
& \mathbf{0}
\\
  \mathbf{0} 
& \Wt
& -\dfrac{1}{\beta} \partial_{c_{i_{2}j_{2}}} \Wt
\\
  \mathbf{0} 
& \mathbf{0}
& \Wt
\end{array}\mright].
\end{equation}
It can then be shown (although, again, omitted for brevity) by induction, that, for all $k\geq 2$,
\begin{equation}
\label{eq:QijTheorem}
\mathbf{Q}_{e_{1} \leadsto e_{2}}^{k}
=  
\mleft[\begin{array}{ccc}
  \Wt^{k} 
& -\dfrac{1}{\beta} \partial_{c_{i_{1}j_{1}}} \Wt^{k} 
& \dfrac{1}{\beta^2}\mathbf{S}_{e_{1} \leadsto e_{2}}^{(k)} 
\\
  \mathbf{0} 
& \Wt^{k}
& -\dfrac{1}{\beta} \partial_{c_{i_{2}j_{2}}}\Wt^{k} 
\\
  \mathbf{0} 
& \mathbf{0}      
& \Wt^{k}
\end{array}\mright].
\end{equation}

Using the above, we can finally write the two-edge likelihood, continuing from Equation~(\ref{eq:TwoEdgeLikelihoodRecall}) as
\begin{equation}
\label{eq:eq:TwoEdgeLikelihoodProofFinal}
  \Lhood(\beta \giv (e_{1}, e_{2}))
= \dfrac{1}{\PF}
  \mathbf{e}_{s}^{\mathsf{T}}
  \sum_{k=2}^{\infty}
  \dfrac{1}{\binom{k}{2}} \mathbf{Q}_{e_{1} \leadsto e_{2}}^{k}
  \mathbf{e}_{2n+t}
= \dfrac{\mleft[ \mathbf{L}_{e_{1} \leadsto e_{2}} \mright]_{s, 2n+t}}{\PF},
\end{equation}
where
\begin{equation}
\label{eq:TwoEdgeMatrixLogarithm}
  \mathbf{L}_{e_{1} \leadsto e_{2}} 
= \sum_{k=2}^{\infty} \dfrac{1}{\binom{k}{2}} \mathbf{Q}_{e_{1} \leadsto e_{2}}^k
\end{equation}
is the two-edge version of $\mathbf{L}_{\tilde{e}}$ defined in Equation~(\ref{eq:OmegaMatrixLogarithm}).

As already mentioned, the proof of Equation~(\ref{eq:MultipleEdgeFullLhood}) for the likelihood for an arbitrary number $M$ of observed edges can be derived similarly to the process presented here.
However, the derivation becomes overly tedious and messy when for an arbitrary $M$ and is left out of the scope of this work.

%
%
\section{Estimation of edge affinities as step selection probabilities}
\label{app:SSPF}
For the landscape graph used in Section~\ref{sec:ApplicationToWildAnimalMovement}, we estimated the edge affinities $a_{ij}$ based on the data and general approach described in~\citep{panzacchi2016predicting} from GPS data for more than 200 wild mountain reindeer from 7 of the largest wild reindeer management areas (including Austhei).
As a slight difference, in~\citep{panzacchi2016predicting} the affinities were based on estimating a Step Selection Function without an intercept, and therefore the model yielded values proportional, but not equal, to the probability of selecting a step. 
For the affinities $a_{ij}$ in this paper, we used the same models, but we refitted them using the method developed in~\citep{lele2009new} to estimate actual probabilities of selection, using Step Selection Probability Functions (SSPF)~\citep{forester2009accounting,lele2006weighted}.
Using the \texttt{ResourceSelection} package \citep{lele2017manual} for R \citep{rcore2015r}, we maximized the following log-likelihood (see \citep{lele2006weighted} for details): 
\begin{equation} \label{eq:likelihood}
\mathcal{L}
(\bm{\beta};
\mathbf{x}_1,
\mathbf{x}_2, \ldots,
\mathbf{x}_m)
=
\sum_{k=1}^K
\log \pi(\mathbf{x_k};\bm{\beta})
-
\log P_k(\bm{\beta})
\end{equation}
where $\mathbf{x}_{k}$ is the vector of $M$ covariates associated to an observed step $k$ between two consecutive locations, $i$ and $j$, $K$ is the number of observed steps in the data set and $\pi$ the probability of selection.
Furthermore,
\begin{equation}
\label{eq:PIntegral}
P_k(\bm{\beta}) = \int \pi(\mathbf{x};\bm{\beta}) f_k (\mathbf{x};s_k) d\mathbf{x},
\end{equation}
where $f_k (\mathbf{x};s_k)$ denotes the distribution of resources available for step $s_k$.
We defined the area within 2 kilometers from the start location $i$ of each step $s_k$ as available. 
We used The Akaike Information Criterion for model selection (see \citep{panzacchi2016predicting} for details). 
Table \ref{tab:suppl_sspf} shows the parameter estimates that maximize the likelihood in Equation (\ref{eq:likelihood}).

\begin{table}[ht]
\footnotesize
\centering
\begin{tabular}{r|rrrr}
 & Estimate & Std.\ Error & $z$-value & Pr($>|z|$) \\ 
  \hline
  Intercept & $-3.79$ & $0.79$ & $-4.77$ & $< 0.001$ \\ 
  Step Length & $-1.24 \times 10^{-3}$ & $2.51\times10^{-5}$ & $-49.53$ & $< 0.001$ \\ 
  (max.\ Slope)$^2$ & $-1.33\times10^{-3}$ & $4.79\times10^{-5}$ & $-27.81$ & $< 0.001$ \\ 
  max.\ Solar Radiation & $0.30$ & $0.02$ & $17.77$ & $< 0.001$ \\ 
  max.\ Trail dens. & $-0.19$ & $0.03$ & $-6.34$ & $< 0.001$ \\ 
  max.\ Road dens. & $0.46$ & $0.27$ & $1.70$ & $0.09$ \\ 
  crossing Road & $-1.19$ & $0.37$ & $-3.20$ & $< 0.001$ \\ 
  prop.\ Bog & $-0.31$ & $0.41$ & $-0.76$ & $0.45$ \\ 
  prop.\ non-Forage & $0.20$ & $0.16$ & $1.21$ & $0.22$ \\ 
  prop.\ Forage & $0.88$ & $0.14$ & $6.37$ & $< 0.001$ \\ 
  prop.\ Lakes & $-1.40$ & $0.33$ & $-4.21$ & $< 0.001$ \\ 
  prop.\ Reservoirs & $-3.97$ & $0.84$ & $-4.70$ & $< 0.001$ \\ 
\end{tabular}
\caption{Fitted coefficients for the summer step selection model of reindeer. ``Max.'' denotes the maximum value along the step, ``prop.'' the proportion of the land cover class along the step, and ``dens.'' is the road density, which is the length of road within a 5 km radius. The methods are detailed in \citep{panzacchi2016predicting}, our only extension was to fit the models using a Step Selection Probability Function (see main text for further details).}
\label{tab:suppl_sspf}
\end{table}
We predicted the affinities $a_{ij}$ as the probability of a step between adjacent pixels $i$ and $j$ using the coefficients from the SSPF:
\begin{equation} \label{eq:s_rein}
a_{ij} = \frac{\exp(\bm{\beta}_{SSPF} \mathbf{x})}{1 + \exp(\bm{\beta}_{SSPF} \mathbf{x})}
\end{equation}
where $\bm{\beta}_{SSPF}$ is a row vector with $M$ elements corresponding to the coefficients from the SSPF (see Table \ref{tab:suppl_sspf}), and $\mathbf{x}$ is a column vector with $M$ elements describing the environmental characteristics of the transition (first element is the intercept, and equals 1). Thus, $a_{ij}$ is the probability of selection of step $i$-to-$j$ (instead of staying put) based on the vector of covariates (e.g. geographic distance between $i$ and $j$, road crossing, proportion of each land cover) characterizing this transition.



\newpage


\bibliographystyle{comnet}
\bibliography{MLE-RSP-biblio}

\end{document}